\documentclass[11pt]{article}

\usepackage{fullpage}

\usepackage{epsfig}
\usepackage{amsfonts}
\usepackage{amssymb}
\usepackage{amstext}
\usepackage{amsmath}

\usepackage{xspace}
\usepackage{color}
\usepackage[section]{algorithm}
\usepackage{algorithmicx}
\usepackage{algpseudocode}
\usepackage{enumitem}
\usepackage{comment}
\usepackage{caption}
\usepackage{subcaption}

\usepackage{multirow}

\usepackage{amsthm, thmtools} 

\usepackage{nameref}
\definecolor{ForestGreen}{rgb}{0.1333,0.5451,0.1333}
\definecolor{DarkRed}{rgb}{0.8,0,0}
\definecolor{Red}{rgb}{0.9,0,0}
\usepackage[linktocpage=true,
	pagebackref=true,colorlinks,
	linkcolor=DarkRed,citecolor=ForestGreen,
	bookmarks,bookmarksopen,bookmarksnumbered]
	{hyperref}
\usepackage{cleveref}

\usepackage{thm-restate} 

\makeatletter
\def\thmt@refnamewithcomma #1#2#3,#4,#5\@nil{%
  \@xa\def\csname\thmt@envname #1utorefname\endcsname{#3}%
  \ifcsname #2refname\endcsname
    \csname #2refname\expandafter\endcsname\expandafter{\thmt@envname}{#3}{#4}%
  \fi
}
\makeatother

\declaretheorem[numberwithin=section,refname={Theorem,Theorems},Refname={Theorem,Theorems}]{theorem}
\declaretheorem[numberlike=theorem,refname={Fact,Facts},Refname={Fact,Facts}]{fact}
\declaretheorem[numberlike=theorem,refname={Lemma,Lemmas},Refname={Lemma,Lemmas}]{lemma}
\declaretheorem[numberlike=theorem,refname={Corollary,Corollaries},Refname={Corollary,Corollaries}]{corollary}
\declaretheorem[name=Theorem,numberlike=theorem,refname={Theorem,Theorems},Refname={Theorem,Theorems}]{proposition} 
\declaretheorem[numberlike=theorem,refname={Observation,Observations},Refname={Observation,Observations}]{observation}

\declaretheorem[numberlike=theorem,refname={Claim,Claims},Refname={Claim,Claims}]{claim}
\declaretheorem[numberlike=theorem]{definition}

\declaretheorem[numberwithin=section,style=problemstyle,name=Problem,refname={Problem,Problems},Refname={Problem,Problems}]{problem}

\newcommand{\poly}{\operatorname{poly}}
\newcommand{\sssp}{{\sf SSSP}\xspace}
\newcommand{\mssp}{{\sf MSSP}\xspace}
\newcommand{\apsp}{{\sf APSP}\xspace}
\newcommand{\mst}{{\sf MST}\xspace}
\newcommand{\mincut}{{\sf mincut}\xspace}
\newcommand{\ksp}{{\sf SP}\xspace}
\newcommand{\asp}{{\sf ASP}\xspace}
\newcommand{\congest}{\ensuremath{{\sf CONGEST}}\xspace}
\newcommand{\congestb}{\ensuremath{{\sf CONGEST}(B)}\xspace}
\newcommand{\local}{\ensuremath{{\sf LOCAL}}\xspace}
\newcommand{\dist}{\operatorname{\sf dist}}
\newcommand{\ecc}{\ensuremath{{\sf D}}\xspace}
\newcommand{\diam}{\ensuremath{{\sf D}}\xspace}

\newcommand{\spdiam}{\ensuremath{{\sf SPDiam}}\xspace}
\newcommand{\spdist}{\ensuremath{{\sf spdist}}\xspace}

\newcommand{\cP}{\mathcal{P}}
\newcommand{\cW}{\mathcal{W}}
\newcommand{\cS}{\mathcal{S}}
\newcommand{\cA}{\mathcal{A}}
\renewcommand{\paragraph}[1]{\medskip\noindent{\bf #1.}\xspace}
\newcommand{\mset}{{\mathcal M}}
\newcommand{\disj}{{\sf disj}\xspace}

\newcommand{\squishlist}{
 \begin{list}{$\bullet$}
  { \setlength{\itemsep}{0pt}
     \setlength{\parsep}{2pt}
     \setlength{\topsep}{2pt}
     \setlength{\partopsep}{0pt}
     \setlength{\leftmargin}{1.5em}
     \setlength{\labelwidth}{1em}
     \setlength{\labelsep}{0.5em} } }
\newcommand{\squishend}{
  \end{list}  }

\ifdefined\ShowComment

\def\danupon#1{\marginpar{$\leftarrow$\fbox{D}}\footnote{$\Rightarrow$~{\sf #1 --Danupon}}}
\def\selfnote#1{\marginpar{$\leftarrow$\fbox{D}}\footnote{$\Rightarrow$~{\sf #1 --Selfnote}}}

\else

\def\danupon#1{}
\def\selfnote#1{}

\fi

\newcommand{\footlabel}[2]{%
    \addtocounter{footnote}{1}%
    \footnotetext[\thefootnote]{%
        \addtocounter{footnote}{-1}%
        \refstepcounter{footnote}\label{#1}%
        #2%
    }%
    $^{\ref{#1}}$%
}

\newcommand{\footref}[1]{%
    $^{\ref{#1}}$%
}

\title{Distributed Approximation Algorithms for Weighted Shortest Paths}
\date{}

\author{Danupon Nanongkai\thanks{ICERM, Brown University, USA. Work partially done while at University of Vienna, Austria and Nanyang Technological University, Singapore, and while supported in part by the following research grants: Nanyang Technological University grant M58110000, Singapore Ministry of Education (MOE) Academic Research Fund (AcRF) Tier 2 grant MOE2010-T2-2-082, and Singapore MOE  AcRF Tier 1 grant MOE2012-T1-001-094.}
}

\begin{document}

\begin{titlepage}
\pagenumbering{roman}
\maketitle
\begin{abstract}

A distributed network is modeled by a graph having $n$ nodes (processors) and  diameter $D$. We study the time complexity of approximating {\em weighted} (undirected) shortest paths on distributed networks with a $O(\log n)$ {\em bandwidth restriction} on edges (the standard synchronous \congest model). The question whether approximation algorithms help speed up the shortest paths (more precisely distance computation) was raised since at least 2004 by Elkin (SIGACT News 2004). The unweighted case of this problem is well-understood while its weighted counterpart is  fundamental problem in the area of distributed approximation algorithms and remains widely open. We present new algorithms for computing both single-source shortest paths (\sssp) and all-pairs shortest paths (\apsp) in the weighted case. 

Our main result is an algorithm for \sssp. Previous results are the classic $O(n)$-time Bellman-Ford algorithm and an $\tilde O(n^{1/2+1/2k}+D)$-time $(8k\lceil \log (k+1) \rceil -1)$-approximation algorithm, for any integer $k\geq 1$, which follows from the result of Lenzen and Patt-Shamir (STOC 2013). (Note that Lenzen and Patt-Shamir in fact solve a harder problem, and we use $\tilde O(\cdot)$ to hide the $O(\poly\log n)$ term.) We present an $\tilde O(n^{1/2}D^{1/4}+D)$-time $(1+o(1))$-approximation algorithm for \sssp. This algorithm is {\em sublinear-time} as long as $D$ is sublinear, thus yielding a sublinear-time algorithm with almost optimal solution. When $D$ is small, our running time matches the lower bound of $\tilde \Omega(n^{1/2}+D)$ by Das Sarma et al. (SICOMP 2012), which holds even when $D=\Theta(\log n)$, up to a $\poly\log n$ factor. 


As a by-product of our technique, we obtain a simple $\tilde O(n)$-time $(1+o(1))$-approximation algorithm for \apsp, improving the previous $\tilde O(n)$-time $O(1)$-approximation algorithm following from the results of Lenzen and Patt-Shamir. We also prove a matching lower bound. Our techniques also yield an $\tilde O(n^{1/2})$ time algorithm on fully-connected networks, which guarantees an exact solution for \sssp and a $(2+o(1))$-approximate solution for \apsp. All our algorithms rely on two new simple tools: {\em light-weight algorithm for bounded-hop \sssp} and {\em shortest-path diameter reduction via shortcuts}. These tools might be of an independent interest and useful in designing other distributed algorithms. 
\end{abstract}

\newpage
\tableofcontents
\end{titlepage}

\newpage
\pagenumbering{arabic}

\section{Introduction}

It is a fundamental issue to understand the possibilities and limitations of distributed/decentralized computation, i.e., to what degree local information is sufficient to solve global tasks. Many tasks can be solved entirely via local communication, for instance, how many friends of friends one has. Research in the last 30 years has shown that  some classic combinatorial optimization problems such as matching, coloring, dominating set, or approximations thereof can be solved   using small (i.e., polylogarithmic)  local communication. 
%
%
However, many important optimization problems  are ``global'' problems from the distributed computation point of view. To count the total number of nodes, to determining the diameter of the system, or to compute a spanning tree, information necessarily must travel to the farthest nodes in a system. If exchanging a  message over a single edge costs one time unit, one needs $\Omega(\diam)$ time units to compute the result, where $\diam$ is the network's diameter. In a model where message size could be unbounded (often known as the \local model), one can simply collect all the information in $O(\diam)$ time (ignoring time for the local computation), and then compute the result. A more realistic model, however, has to take into account the congestion issue and limits the size of a message allowed to be sent in a single communication round to some $B$ bits, where $B$ is typically set to $\log n$. This model is often called synchronous \congest (or $\congest(B)$ if $B\neq \log n$). 
%
%
%
Time complexity in this model is one of the major studies in distributed computing \cite{peleg}.

Many previous works in this model, including several previous FOCS/STOC papers (e.g. \cite{GarayKP98,PelegR00,Elkin06,DasSarmaHKKNPPW11,LenzenP13}), concern {\em graph problems}. Here, we want to learn some topological properties of a network, such as minimum spanning tree (\mst), minimum cut (\mincut), and distances. 
These problems can be trivially solved in $O(m)$ rounds, where $m$ is the number of edges, by aggregating the whole network into one node. Of course, this is neither interesting nor satisfactory. The holy grail in the area of distributed graph algorithms is to beat this bound and, in many case, obtain a {\em sublinear-time} algorithm whose running time is in the form $\tilde O(n^{1-\epsilon}+\diam)$ for some constant $\epsilon>0$, where $n$ is the number of nodes and $\diam$ is the network's diameter.
For example, through decades of extensive research, we now have an algorithm that can find an \mst in  $\tilde O(n^{1/2}+\diam)$ time \cite{GarayKP98,KuttenP98}, and we know that this running time is tight \cite{PelegR00}. This algorithm serves as a building block for several other sublinear-time algorithms (e.g. \cite{Thurimella97,PritchardT11,GhaffariK13,Nanongkai14_podc,Su14}). 

It is also natural to ask whether we can further improve the running time of existing graph algorithms by mean of {\em approximation}, e.g., if we allow an algorithm to output a spanning tree that is almost, but not, minimum. This question has generated a research in the direction of {\em distributed approximation algorithms} which has become fruitful in the recent years.
%
On the negative side, Das Sarma et al. \cite{DasSarmaHKKNPPW11} (building on \cite{PelegR00,Elkin06,KorKP13}) show that \mst and a dozen other problems, including \mincut and computing the distance between two nodes, cannot be computed faster than $\tilde O(n^{1/2}+\diam)$ in the synchronous \congest model even when we allow a large approximation ratio. On the positive side, we start to be able to solve some problems in sublinear time by sacrifying a small approximation factor; e.g., we can $(1+\epsilon)$-approximate \mincut in $\tilde O(n^{1/2}+\diam)$ time \cite{Nanongkai14_podc,Su14} and $(3/2)$-approximate the network's diameter in $O(n^{3/4}+\diam)$ time \cite{HolzerW12,PelegRT12}. 
%
The question whether distributed approximation algorithms can help improving the time complexity of computing shortest paths was raised a decade ago by Elkin \cite{Elkin04}. 
It is surprising that, despite so much progress on other problems in the last decade, the problem of computing shortest paths is still widely open, especially when we want a small approximation guarantee. Prior to our work, sublinear-time algorithms for computing single-source shortest path (\sssp) and linear-time algorithms for computing all-pairs shortest paths (\apsp) have to pay a high approximation factor \cite{LenzenP13}. This paper fills this gap with algorithms having small approximation guarantees.


\subsection{The Model}  \label{sec:model}

Consider a network of processors  modeled by an undirected unweighted  $n$-node $m$-edge 
graph $G$, where nodes model the processors and  edges model the {\em bounded-bandwidth} links between the processors. Let $V(G)$ and $E(G)$ denote the set of nodes and edges of $G$, respectively. 
%
%
%
The processors  (henceforth, nodes) are assumed to have unique IDs in the range of $\{1, \ldots, \poly(n)\}$ and infinite computational power. Each node has limited topological knowledge; in particular, it only knows the IDs of its neighbors and knows {\em no} other topological information (e.g., whether its neighbors are linked by an edge or not). Nodes may also accept some additional inputs as specified by the problem at hand.

For the case of graph problems, the additional input is {\em edge weights}.  Let $w:E(G)\rightarrow \{1, 2, \ldots, \poly(n)\}$ be the edge weight assignment. We refer to network $G$ with weight assignment $w$ as the {\em weighted network}, denoted by $(G, w)$. The weight $w(uv)$ of each edge $uv$ is known only to $u$ and $v$. As commonly done in the literature (e.g., \cite{KhanP08,LotkerPR09,KuttenP98,GarayKP98,GhaffariK13}), we will assume that the maximum weight is $\poly(n)$; so, each edge weight can be sent through an edge (link) in one round.\footnote{We note that, besides needing this assumption to ensure that weights can be encoded by $O(\log n)$ bits, we also need it in the analysis of the running time of our algorithms: most running times of our algorithms are logarithmic of the largest edge weight. This is in the same spirit as, e.g., \cite{LotkerPR09,GhaffariK13,KhanP08}.}

%


There are several measures to analyze the performance of such algorithms, a fundamental one being the running time, defined as the worst-case number of {\em rounds} of distributed communication. 
%
%
At the beginning of each round, all nodes wake up simultaneously. Each node $u$ then sends an arbitrary message of $B=\log n$ bits through each edge $uv$, and the message will arrive at node $v$ at the end of the round. 


%

We assume that nodes always know the number of the current round. To simplify notations, we will name nodes using their IDs, i.e. we let $V(G)\subseteq \{1, \ldots, \poly(n)\}$. Thus, we use $u\in V(G)$ to represent a node, as well as its ID. The running time is analyzed in terms of number of nodes $n$, number of edge $m$, and $\diam$, the diameter of the network $G$. Since we can compute $n$ and $2$-approximate $\diam$ in $O(\diam)$ time, we will assume that every node knows $n$ and the $2$-approximate value of $\diam$. We say that an event holds {\em with high probability} (w.h.p.) if it holds with probability at least $1-1/n^c$, where $c$ is an arbitrarily large constant.

\subsection{Problems \& Definitions}


For any nodes $u$ and $v$, a {\em $u$-$v$ path} $P$ is a path $\langle u=x_0, x_1, \ldots, x_\ell=v\rangle$ where $x_ix_{i+1}\in E(G)$ for all $i$. For any weight assignment $w$, we define the weight or distance of $P$ as $w(P)=\sum_{i=0}^{\ell-1} w(x_ix_{i+1})$. 
Let $\cP_G(u,v)$ denote the set of all $u$-$v$ paths in $G$. We use $\dist_{G, w}(u, v)$ to denote the distance from $u$ to $v$ in $(G,  w)$; i.e., $\dist_{G,  w}(u, v)=\min_{P\in \cP_G(u, v)}  w(P)$. We say that a path $P$ is a {\em shortest $u$-$v$ path} in $(G,  w)$ if $ w(P)=\dist_{G,  w}(u, v)$.
The {\em diameter} of $(G,  w)$ is $D(G,  w)=\max_{u, v} \dist_{G,  w}(u, v)$. 
%
When we want to talk about the properties of the underlying undirected unweighted network $G$, we will drop $ w$ from the notations. Thus, $\dist_G(u, v)$ is the distance between $u$ and $v$ in $G$ and,  $\diam(G)$ is the diameter of $G$. We refer to $\diam(G)$ by ``hop diameter'', or sometimes simply ``diameter'', and $\diam(G, w)$ by ``weighted diameter''. When it is clear from the context, we use $\diam$ to denote $\diam(G)$. We emphasize that, like other papers in the literature, the term $\diam$ which appears in the running time of our algorithms is the diameter of the underlying {\em unweighted} network $G$.


%
\begin{definition}[Single-Source and All-Pairs Shortest Paths (\sssp, \apsp)]\label{def:sssp}\label{def:apsp} 
In the {\em single-source shortest paths problem (\/\sssp)}, we are given a weighted network $(G, w)$ as above and a {\em source} node $s$ (the ID of $s$ is known to every node). We want to find the distance between $s$ and every node $v$ in $(G, w)$, denoted by $\dist_{G, w}(s, v)$. In particular, we want $v$ to know the value of $\dist_{G, w}(s, v)$.
In the all-pairs shortest paths problem (\apsp), we want to find $\dist_{G, w}(u, v)$ for every pair $(u, v)$ of nodes. In particular, we want both $u$ and $v$ to know the value of $\dist_{G, w}(u, v)$.
\end{definition}
%
For any $\alpha$, we say an algorithm $\cA$ is an {\em $\alpha$-approximation} algorithm for \sssp if  it outputs $\widetilde{\dist}_{G, w}$ such that $\dist_{G, w}(s, v) \leq \widetilde{\dist}_{G, w}(s, v)\leq \alpha\dist_{G, w}(s, v)$ for all $v$. 
Similarly, we say that $\cA$ is an $\alpha$-approximation algorithm for \apsp if it outputs $\widetilde{\dist}_{G, w}$ such that $\dist_{G, w}(u, v) \leq \widetilde{\dist}_{G, w}(u, v)\leq \alpha\dist_{G, w}(u, v)$ for all $u$ and $v$. 

\paragraph{Remark} We emphasize that we do {\em not} require every node to know {\em all} distances. Also note that, while our paper focuses on computing distances between nodes, it can be used to find a routing path or compute a routing table as well. For example, after solving \apsp, nodes can exchange all distance information with their neighbors in $O(n)$ time. Then, when a node $u$ wants to send a message to node $v$, it simply sends such message to the neighbor $x$ with smallest $\dist_{G, w}(v, x)$. 
The name shortest paths is inherited from \cite{DasSarmaHKKNPPW11} (see the definition of shortest $s$-$t$ path problem in \cite[Section 2.5]{DasSarmaHKKNPPW11}) and in particular the lower bound in \cite{DasSarmaHKKNPPW11} holds for our problem).

\begin{table*}
\centering
\begin{tabular}{ |l|l|l|l|l| }
\hline
Problems & Topology & References & Time & Approximation \\ 
\hline
\multirow{5}{*}{\sssp} & \multirow{3}{*}{General} & Bellman\&Ford \cite{Bellman58,Ford56} & $\tilde O(n)$ & exact \\
& & Lenzen\&Patt-Shamir \cite{LenzenP13}\footref{christoph}  & $\tilde O(n^{1/2+1/2k}+\diam)$ & $8k\lceil \log (k+1) \rceil -1$ \\
 & & {\bf this paper} & $\tilde O(n^{1/2}\diam^{1/4}+\diam)$ & $1+o(1)$ \\
 & & & ($= \tilde O(n^{3/4}+\diam)$) & \\
 \cline{2-5}
& \multirow{2}{*}{Fully-Connected} & Baswana\&Sen \cite{BaswanaS07} & $\tilde O(n^{1/k})$ & $2k-1$ \\
& & {\bf this paper} & $\tilde O(n^{1/2})$ & exact \\
\hline
\multirow{5}{*}{\apsp} & \multirow{3}{*}{General} & Trivial & $O(m)$ & exact \\
& & Lenzen\&Patt-Shamir \cite{LenzenP13} & $\tilde O(n)$ & $O(1)$ \\
 & & {\bf this paper} & $\tilde O(n)$ & $1+o(1)$ \\
 \cline{2-5}
& \multirow{2}{*}{Fully-Connected} & Baswana\&Sen \cite{BaswanaS07} & $\tilde O(n^{1/k})$ & $2k-1$ \\
& & {\bf this paper} & $\tilde O(n^{1/2})$ & $2+o(1)$ \\
\hline
\end{tabular}
\caption{Summary of previous and our results (presented in \Cref{sec:related,sec:results}). Parameter $k\geq 1$ is an integer (note that the time complexities above are sublinear only when $k\geq 2$). Note that this table omits previous running times based on other parameters (such as shortest-path diameter).}\label{table:results summary}
\end{table*}

\subsection{Our Results}\label{sec:results}
Our and previous results are summarized in \Cref{table:results summary} (see \Cref{sec:related} for the details of previous results). As shown in the table, previous algorithms either have large approximation guarantee or large running time. In this paper, we aim at algorithms with both small approximation guarantees and small running time. We consider both \sssp and \apsp and study algorithms on both general networks and fully-connected networks.
Our main result is a sublinear-time $(1+o(1))$-approximation algorithm for \sssp on general graphs:
\begin{theorem}[\sssp on general graph]\label{thm:main sssp}
There is a distributed $(1+o(1))$-approximation algorithm that solves \sssp on any weighted $n$-node network $(G, w)$. It finishes in  $\tilde O(n^{1/2}\diam^{1/4}+\diam)$ time w.h.p.
\end{theorem}


For typical real-world networks (e.g., ad hoc networks and peer-to-peer networks) $\diam$ is {\em small} (usually $\tilde O(1)$). (In some networks, an even stronger property also holds; e.g., a peer-to-peer network is usually assumed to be an expander \cite{AugustinePRU12}.) It is thus of a special interest to develop an algorithm in this setting. For example, \cite{LotkerPP06} studied \mst on constant-diameter networks. Das Sarma et al. \cite{DasSarmaNPT13} developed a $\tilde O((\ell \diam)^{1/2})$-time algorithm for computing a random walk of length $\ell$, which is faster than the trivial $O(\ell)$-time algorithm when $\diam$ is small.\danupon{We may also mention $\tilde O(n^{1/2}D)$-time algorithm for approximating the diameter \cite{PelegRT12}} In the same spirit, our algorithm is faster than previous algorithms. Moreover, in this case our running time matches the lower bound of $\tilde \Omega(n^{1/2}+\diam)$ \cite{DasSarmaHKKNPPW11,ElkinKNP12}, which holds even for any algorithm with $\poly(n)$ approximation ratio; thus, our result settles the status of \sssp for this case. 
%
Additionally, since the same lower bound also holds in the quantum setting \cite{ElkinKNP12}, our result makes \sssp among a few problems (others are \mst and \mincut) that quantum communication cannot help speeding up distributed algorithms significantly.

Observe that our running time is {\em sublinear} as long as $\diam$ is sublinear in $n$ (since $\tilde O(n^{1/2}\diam^{1/4}+\diam)$ can be written as $\tilde O(n^{3/4}+\diam)$). As shown in \Cref{table:results summary}, previously we can either solve \sssp exactly in $\tilde O(n)$ time using Bellman-Ford algorithm \cite{Bellman58,Ford56} or $(8k\lceil \log (k+1) \rceil -1)$-approximately, for any $k>1$, in  $\tilde O(n^{1/2+1/2k}+\diam)$ time by applying the algorithm of Lenzen and Patt-Shamir \cite{LenzenP13}\footlabel{christoph}{Note that by applying the technique of Lenzen and Patt-Shamir with carefully selected parameters, the approximation ratio can be reduced to $4k-1$. We thank Christoph Lenzen (personal communication) for this information.}. Our algorithm is the first that gives an output very close to the optimal solution in sublinear time.
%
%
%
Our result also points to an interesting direction in proving a stronger lower bound for \sssp: in contrast to previous lower bound techniques which usually work on low-diameter networks, proving a stronger lower bound for \sssp needs a new technique that must exploit the fact that the network's diameter is fairly large.\danupon{Should we also mention that we can also get $(2k-1)$-approximation?}
%
%
As a by-product of our techniques, we also obtain a linear-time algorithm for \apsp. 
\begin{theorem}[\apsp on general graphs]\label{thm:linear apsp}
There is a distributed $(1+o(1))$-approximation algorithm that solves \apsp on any weighted $n$-node network $(G, w)$ which finishes in $\tilde O(n)$ time w.h.p.
\end{theorem}
We also observe that this algorithm is essentially {\em tight}:
%
%
\begin{restatable}[Lower bound for \apsp]{observation}{lowerbound}\label{observe:lower bound}
Any $\poly(n)$-approximation algorithm for \apsp on an $n$-node weighted network $G$ requires $\Omega(\frac{n}{\log n})$ time. This lower bound holds even when the underlying network $G$ has diameter $\diam(G)=2$.
Moreover, for any $\alpha(n)=O(n)$, any $\alpha(n)$-approximation algorithm on an {\em unweighted} network requires $\Omega(\frac{n}{\alpha(n)\log n})$ time. 
\end{restatable}
\Cref{observe:lower bound} implies that the running time of our algorithm in \Cref{thm:linear apsp} is {\em tight} up to a $\poly\log n$ factor, unless we allow a prohibitively large approximation factor of $\poly n$. Moreover, even when we restrict ourselves to unweighted networks, we still cannot significantly improve the running time, unless the approximation ratio is fairly large; e.g., any $n^{1-\delta}$-time algorithm must allow an approximation ratio of $\Omega(n^\delta/\log n)$.
We note that a similar result to \Cref{observe:lower bound} has been independently proved by Lenzen and Patt-Shamir \cite{LenzenP13} in the context of name-independent routing scheme.

Other by-products of our techniques are efficient algorithms on fully-connected distributed networks, i.e., when $G$ is a complete graph. As mentioned earlier, it is of an interest to study algorithms on low-diameter networks. The case of fully-connected networks is an extreme case where $\diam=1$. This special setting captures, e.g., overlay and peer-to-peer networks, and has received a considerable attention recently (e.g. \cite{LotkerPPP03,LenzenW11,PattShamirT11,Lenzen13,DolevLP12tri,BernsHP12}). Obviously, this model gives more power to algorithms since every node can directly communicate with all other nodes; for example, MST can be constructed in $O(\log\log n)$ time \cite{LotkerPPP03}, as opposed to the $\tilde \Omega(n^{1/2}+\diam)$ lower bound on general networks. No sublinear-time algorithm for \sssp and \apsp is known even on this model if we want an optimal or near-optimal solution. 
%
%
In this paper, we show such an algorithm. First, note that our $\tilde O(n^{1/2}\diam^{1/4}+\diam)$-time algorithm in \Cref{thm:main sssp} already implies that \sssp can be $(1+o(1))$-approximated in  $\tilde O(n^{1/2})$ time. We show that, as an application of our techniques for proving \Cref{thm:main sssp}, we can get an {\em exact} algorithm within the same running time. More importantly, we show that these techniques, combined with some new ideas, lead to a $(2+o(1))$-approximation $\tilde O(\sqrt{n})$-time algorithm for \apsp. The latter result is in contrast with the general setting where we show that a sublinear running time is impossible even when we allow large approximation ratios (\Cref{observe:lower bound}). 
\begin{theorem}[Sublinear time algorithm on fully-connected networks]
On any fully-connected weighted network, in $\tilde O(n^{1/2})$ time, \sssp can be solved exactly and \apsp can be $(2+o(1))$-approximated w.h.p. 
\end{theorem}




%
%
%

\subsection{Related Work}\label{sec:related}

\paragraph{Unweighted Case} \sssp and \apsp are essentially well-understood in the unweighted case. \sssp can be trivially solved in $O(\diam)$ time using a breadth-first search tree \cite{peleg,Lynch-Book}. 
%
%
Frischknecht, Holzer, and Wattenhofer \cite{FrischknechtHW12,HolzerW12} show a (surprising) lower bound of $\Omega(n/\log n)$ for computing the diameter of unweighted networks, which implies a lower bound for solving unweighted \apsp. This lower bound holds even for $(3/2-\epsilon)$-approximation algorithms. This lower bound is matched (up to a $\poly\log n$ factor) by  $O(n)$-time deterministic {\em exact} algorithms for unweighted \apsp found independently by \cite{HolzerW12} and \cite{PelegRT12}.
Another case that has been considered is when nodes can talk to any other node in one time unit. This can be thought of as a special case of \apsp on fully-connected networks where edge weights are either $1$ or $\infty$. In this case, Holzer \cite{Holzer13thesis} shows that \sssp can be solved in $\tilde O(n^{1/2})$ time\footnote{We thank Stephan Holzer for pointing this out.}. 

\paragraph{Name-Dependent Routing Scheme} 
%
For the weighted \sssp and \apsp on general networks, the best known results follow from the recent algorithm for computing tables for {\em name-dependent routing} and {\em distance approximation} by Lenzen and Patt-Shamir \cite{LenzenP13}.
In particular, consider any integer $k>1$. Lenzen and Patt-Shamir \cite[Theorem 4.12]{LenzenP13} showed that in time $\tau=\tilde O(n^{1/2+1/2k}+\diam)$ every node $u$ can compute a {\em label} $\lambda(u)$ of size $\sigma=O(\log (k+1)\log n)$ and a function $dist_u$ that maps label $\lambda(v)$ of any  node $v$ to a distance approximation $dist_u(v)$ such that $\dist_{G, w}(u, v) \leq dist_u(v) \leq \rho\dist_{G, w}(u, v)$ where $\rho=8k\lceil \log (k+1)\rceil -1$.\footnote{We note that Lenzen and Patt-Shamir \cite[Theorem 4.12]{LenzenP13} actually allow $k=1$. However, their algorithm relies on the result of Baswana and Sen (see Theorem 4.7 in their paper) which does not allow $k=1$.}
\danupon{Theorem 4.19 in the arXiv version and 4.12 in the STOC version of \cite{LenzenP13}: Consider any $\alpha=1/2+\epsilon$ for some $1/2\leq \epsilon\geq 1/\log n$ (see the paper for the case where $\epsilon$ is very close to $0$ which will not be considered here). Define $k=\lceil 1/2\epsilon\rceil$. We get stretch $\rho(\alpha)=8k\lceil \log (k+1)\rceil -1\rceil$, label size $O(\log (k+1)\log n)$, time $\tilde O(n^{\alpha}+\diam)$}
%
%
We can solve \sssp by running the above algorithm of Lenzen and Patt-Shamir and broadcasting the label $\lambda(s)$ of the source $s$ to all nodes. This takes time $\tau+\sigma = \tilde O(n^{1/2+1/2k}+\diam)$ and has an approximation guarantee of $\rho$. We can solve \apsp by broadcasting $\lambda(v)$ for all $v$, taking time $\tau+n\sigma = \tilde O(n)$.

\paragraph{Sparsification} 
The shortest path problem is one of the main motivations to study distributed algorithms for graph sparsification.  
%
%
These algorithms\footnote{We note that some of these algorithms (e.g., \cite{Elkin05,KhanKMPT12}) can actually solve a more general problem called the {\em $\cS$-shortest path} problem. To avoid confusions, we will focus only on \sssp and \apsp.} have either super-linear time or large approximation guarantees.
For example, 
Elkin and Zhang \cite{ElkinZ06} present an algorithm for the unweighted case based on a sparse spanner that takes (very roughly) $O(n^{\xi})$ time and gives $(1+\epsilon)$-approximate solution, for small constants $\xi$ and $\epsilon$. The algorithm is also extended to the weighted case but both running time and approximation guarantee are large (linear in terms of the largest edge weight).
The running time could be traded-off with the approximation guarantee using, e.g., a $(2k-1)$-spanner of size $O(kn^{1+1/k})$ \cite{BaswanaS07} where $k$ can vary; e.g., by setting $k=\log n/\log\log n$, we have an $O(n\log n)$-time $O(\log n/\log \log n)$-approximation algorithm (we need $O(k^2)$ to construct a spanner and $O(kn^{1+1/k})=O(n\log n)$ to aggregate it). The spanner of \cite{Pettie10} can also be used to get a linear-time $(2^{O(\log^*{n})}\log n)$-approximation algorithm in the unweighted case.

\danupon{According to STOC'14 committee, \cite{DerbelGPV08} (which I used to cite) needs large messages so it's quite irrelevant.}

In general, it is not clear how to use graph sparsification for computing shortest paths since we still need at least linear time to collect the sparse graph. However, it plays a crucial role in some previous algorithms, including the algorithm of Lenzen and Patt-Shamir \cite{LenzenP13}. Moreover, by running the graph sparsification algorithm of Baswana and Sen \cite{BaswanaS07} and collecting the network to one node, we can $(2k-1)$-approximate \apsp in $O(kn^{1+1/k})$ time on general networks and $O(kn^{1/k})$ time on fully-connected networks, for any integer $k\geq 2$. This gives the fastest algorithm (with high approximation guarantees) on fully-connected networks.


\paragraph{Other Parameters} There are also some approximation algorithms whose running time is based on other parameters. These algorithms do not give any improvement for the worst values of their parameters. We do not consider these parameters in this paper since they are less standard. 
One important parameter is the {\em shortest-path diameter}, denoted by $\spdiam(G, w)$. This parameter captures the number of edges in a shortest path between any pair of nodes (see \Cref{def:spdiam} for details). It naturally arises in the analysis of several algorithms. For example, Bellman-Ford algorithm \cite{Bellman58,Ford56} can be analyzed to have $O(\spdiam(G, w))$ time for \sssp. 
Khan et al. \cite{KhanKMPT12} gives a $\tilde O(n\cdot \spdiam(G, w))$-time $O(\log n)$-approximation algorithm via {\em metric tree embeddings}~\cite{FakcharoenpholRT04}. We can also construct {\em Thorup-Zwick distance sketches}~\cite{ThorupZ05} of size $O(kn^{1/k})$ and stretch $2k-1$ in $\tilde O(kn^{1/k}\cdot \spdiam(G, w))$ time~\cite{DasSarmaDP12}. Since $\spdiam(G, w)$ can be as large as $n$, these algorithms do not give any improvement to previous algorithms when analyzed in terms of $n$ and $\diam$. One crucial component of our algorithms involves reducing the shortest-path diameter to be much less than $n$ (more in \Cref{sec:techniques}).
Another shortest path algorithm with running time based on the network's {\em local path diameter} is developed as a subroutine of the approximation algorithm for \mst \cite{KhanP08}. This algorithm solves a slightly different problem (in particular, nodes only have to know the distance to some nearby nodes) and cannot be used to solve \sssp and \apsp.

\paragraph{Lower Bounds} The lower bound of Das Sarma et al. \cite{DasSarmaHKKNPPW11} (building on \cite{Elkin06,PelegR00,KorKP13}) shows that solving \sssp requires $\tilde \Omega(\sqrt{n}+\diam)$ time, even when we allow $\poly(n)$ approximation ratio and the network has $O(\log n)$ diameter. This implies the same lower bound for \apsp. Recently, \cite{ElkinKNP12} shows that the same $\tilde \Omega(\sqrt{n}+\diam)$ lower bound holds even in the quantum setting. These lower bounds are subsumed by Observation~\ref{observe:lower bound} for the case of \apsp. 
Das Sarma et al. (building on \cite{LotkerPP06}) also shows a polynomial lower bound on networks of diameter 3 and 4. It is still open whether there is a non-trivial lower bound on networks of diameter one and two \cite{Elkin04}. 




\paragraph{Other Works} While computing shortest paths is among the earliest studied problems in distributed computing, many classic works on this problem concern other objectives, such as the message complexity and convergence. When faced with the bandwidth constraint, the time complexities of these algorithms become higher than the trivial $O(m)$-time algorithm; e.g., Bellman-Ford algorithm and algorithms in \cite{AbramR82,Haldar97,AfekR93} require $\Omega(n^2)$ time.

To the best of our knowledge, there is still no exact distributed algorithm for \apsp that is faster than the trivial $O(m)$-time algorithm\footnote{The problem can also be solved by running the distributed version of Bellman-Ford algorithm \cite{peleg,Lynch-Book,Santoro-book} from every node, but this takes $O(n^2)$ time in the worst case. So this is always worse than the trivial algorithm.}, except for the special case of {\em BHC network}, whose topology is structured as a balanced hierarchy of clusters. In this special case, the problem can be solved in $O(n \log n)$-time \cite{AntonioHT92}. 
For the related problem of computing network's diameter and girth, many results are known in the unweighted case but none is previously known for the weighted case. Peleg, Roditty, and Tal \cite{PelegRT12} shows that we can $\frac{3}{2}$-approximate the network's diameter in $O(n^{1/2}\diam)$ time, in the unweighted case, and Holzer and Wattenhofer \cite{HolzerW12} presents an $O(\frac{n}{\diam}+\diam)$-time $(1+\epsilon)$-approximation algorithm. By combining both algorithms, we get a $\frac{3}{2}$-approximation  $O(n^{3/4}+\diam)$-time algorithm. In contrast, any $(\frac{3}{2}-\epsilon)$-approximation and $(2-\epsilon)$-approximation algorithm for computing the network's diameter and girth requires  $\Omega(n/\log n)$ time~\cite{HolzerW12}
and $\Omega(\sqrt{n}/\log n)$ time~\cite{FrischknechtHW12}, respectively. These bounds imply the same lower bound for approximation algorithms for \apsp on unweighted networks. In particular, they imply that our approximation algorithms are tight, even on unweighted networks.

\section{Overview}\label{sec:techniques}

\subsection{Tool 1: Light-Weight Bounded-Hop \sssp (Details in \Cref{sec:light-weight bounded-hop sssp})} At the core of our algorithms is the {\em light-weight} $(1+o(1))$-approximation algorithm for computing {\em bounded-hop} distances. Informally, an {\em $h$-hop path} is a path containing at most $h$ edges.  The $h$-hop distance between two nodes $u$ and $v$, denoted by $\dist_{G, w}^h(u, v)$, is the minimum weight among all $h$-hop paths between $u$ and $v$. The {\em $h$-hop \sssp problem} is to find the $h$-hop distance between a given source node $s$ and all other nodes. This problem can be solved exactly in $O(h)$ time using the distributed version of Bellman-Ford algorithm. This algorithm is, however, {\em not} suitable for parallelization, i.e. when we want to solve $h$-hop \sssp from $k$ different sources. The reason is that Bellman-Ford algorithm is {\em heavy-weight} in the sense that they require so much communication between each neighboring nodes; in particular, this algorithm may require as many as $O(h)$ messages on each edge. Thus, running $k$ copies of this algorithm in parallel may require as many as $O(hk)$ messages on each edge, which will require $O(hk)$ time. 

We show a simple algorithm that is not as accurate as Bellman-Ford algorithm but more suitable for parallelization: it can $(1+o(1))$-approximate $h$-hop \sssp in $\tilde O(h)$ time and is {\em light-weight} in the sense that every node sends a message (of size $O(\log n)$) to its neighbors only $\tilde O(1)$ times. Thus, when we run $k$ copies of this algorithm in parallel, we will require to send only $\tilde O(k)$ messages through each edge, which gives us a hope that we will require only additional $\tilde O(k)$ time. By a careful paralellization (based on the random delay technique of \cite{LeightonMR94}\footnote{Note that the random delay technique makes the algorithm randomized. Techniques in \cite{HolzerW12,PelegRT12} might enable us to get a deterministic algorithm. We do not discuss these techniques here since other parts of our algorithms will also heavily rely on randomness.}), we can solve $h$-hop \sssp from $k$ sources in $\tilde O(h+k)$ time. This is the first tool that we will use later.
\begin{claim}[See \Cref{thm:k-source h-hop sp} for a formal statement]\label{claim:k-source h-hop sp}
We can $(1+o(1))$-approximate $h$-hop \sssp from any $k$ nodes in $\tilde O(h+k)$ time. 
\end{claim}
The idea behind \Cref{claim:k-source h-hop sp} is actually very simple. Consider any path $P$ having at most $h$ hops. Let $\epsilon=1/\log n$ and $W'=(1+\epsilon)^i$ where $i$ is such that $W'\leq w(P) \leq (1+\epsilon)W'$ (recall that $w(P)$ is the sum of weights of edges in $P$). Consider changing weight $w$ slightly to $w'$ where $w'(uv)=\lceil \frac{h w(uv)}{\epsilon W'}\rceil$. Because $w'(uv)-\frac{hw(uv)}{\epsilon W'}\leq 1$, we have that  
$$w(uv)\leq w'(uv)\times \frac{\epsilon W'}{h} \leq w(uv)+O(\epsilon)\frac{W'}{h}.$$ 
It follows that  
$$w(P)\leq w'(P)\times \frac{\epsilon W'}{h} \leq w(P)+ O(\epsilon) W' = (1+o(1))w(P).$$
%
%
%
In other words, it is sufficient for us to find $w'(P)$. To this end, we observe that $w'(P)=O(h/\epsilon)$. Thus, we can simply use the breadth-first search (BFS) algorithm \cite{peleg,Lynch-Book} on $(G, w')$ for $O(h/\epsilon)$ rounds. The BFS algorithm is {\em light-weight}: it sends at most one message through each edge. 
Now to use this algorithm to solve $h$-hop \sssp, we have to try different values of $W'$ in the form $(1+\epsilon)^i$. This makes our algorithm send $\tilde O(1)$ messages through each edge.

To the best of our knowledge, this simple technique has not been used before in the literature of distributed algorithms. In the dynamic data structure context, Bernstein has independently used a similar weight rounding technique to construct a {\em bounded-hop data structure}, which plays an important role in his recent breakthrough \cite{Bernstein13}.
Also, it was recently pointed out to us by a STOC 2014 reviewer that this technique is similar to the one used in the PRAM algorithm of Klein-Sairam \cite{KleinS92} which was originally proposed for VLSI routing by Raghavan and Thomson \cite{RaghavanT85}.  The main difference between this and our weight approximation technique is that we always round edge weights up while the previous technique has to round the weights up and down randomly (with some appropriate probability). So, if we adopt the previous technique, then the approximation guarantee of our light-weight \sssp algorithm will hold only with high probability (in contrast, it always holds in this paper). More importantly, randomly rounding the weight could cause some edge to have a zero weight after rounding. This problem can be handled in the PRAM setting by contracting edges of weight zero. However, this will be a serious problem for us since we do not know how to handle zero edge weight.

\subsection{Tool 2: Shortest-Path Diameter Reduction Using Shortcuts (Details in \Cref{sec:spd reduction})} The other crucial idea that we need is the {\em shortest-path diameter reduction technique}. Recall that the shortest-path diameter of a weighted graph $(G, w)$, denoted by $\spdiam(G, w)$, is the minimum number $h$ such that for any nodes $u$ and $v$, there is a shortest $u$-$v$ path in $(G, w)$ having at most $h$ edges; in other words, $\dist_{G, w}^h(u, v)=\dist_{G, w}(u, v)$ for all $u$ and $v$.  
As discussed in \Cref{sec:related} there are algorithms that need $\tilde O(\spdiam(G, w))$ time to solve \sssp and \apsp, e.g. Bellman-Ford algorithm. 
%
%
Thus, it is intuitively important to try to make the shortest-path diameter small. The second crucial tool of our algorithm is the following claim. 
\begin{claim}[See \Cref{thm:spd reduction} for a formal statement]\label{claim:spd reduction}
If we add $k$ edges called {\em shortcuts} from every node $u$ to its $k$ nearest nodes (breaking tie arbitrarily), where for each such node $v$ the shortcut edge $uv$ has weight $\dist_{G, w}(u, v)$, then we can bound the shortest-path diameter to $O(n/k)$.
\end{claim}
We note that the above claim would be trivially true if we add a shortcut from every node to all nodes within $k$ hops from it. The non-trivial part is showing that it is sufficient to add shortcuts to only $k$ nearest nodes. Note that this claim holds only for undirected graphs and the proof has to carefully exploit the fact that the network is undirected. 


To the best of our knowledge, there is no previous work that proves and uses this fact in the distributed setting. 
Previous work that is somewhat related is the BSP algorithm of Lenzen and Patt-Shamir \cite{LenzenP13} which finds $h$-hop distances to $k$ nearest nodes in $O(hk)$ time. In this work, the algorithm is not used to create shortcuts, but rather to collect information about a sufficient number of nodes so that one of them is also in some set of uniformly sampled nodes.
%
%
%
%
%
Another related work is the notion of {\em $(d, \epsilon)$-hop set} introduced by Cohen \cite{Cohen00} in the PRAM setting: 
%
%
our shortest path diameter reduction technique can be considered as a simple construction of $(d, 0)$-hop set of size $O(n^2/d)$. It might be possible to improve our algorithm by applying a more advanced construction of such hop set to the distributed setting.


\subsection{Sketches of Algorithms}

\paragraph{\apsp on General Networks (details in \Cref{sec:sssp})} Algorithm for \apsp follows almost immediately from the the first tool above. By applying \Cref{claim:k-source h-hop sp} with $h=k=n$, we can $(1+o(1))$-approximate \sssp with every node as a source in $\tilde O(n)$ time; in other words, we can $(1+o(1))$-approximate \apsp in $\tilde O(n)$ time on general networks.

\paragraph{\sssp on Fully-Connected Networks (details in \Cref{sec:sssp clique})} This result follows easily from the the second tool above. To compute \sssp exactly on fully-connected networks, we will compute $k$ shortcuts from every node, where $k=n^{1/2}$. To do this, we show that it is enough for every node to send $k$ lightest-weight edges incident to it to all other nodes (since running $k$ rounds of Dijkstra's algorithm will only need these edges). This takes $O(n^{1/2})$ time. 
%
Using this information to modify the weight assignment from $w$ to $w'$, we can reduce the shortest-path diameter of the network to $\spdiam(G, w')\leq \sqrt{n}$ without changing the distance between nodes; this fact is due to \Cref{claim:spd reduction}. We then run Bellman-Ford algorithm on this $(G, w')$ to solve \sssp; this takes  $\spdiam(G, w')=O(n^{1/2})$ time.

\paragraph{\apsp on Fully-Connected Networks (details in \Cref{sec:apsp clique})} We will need both tools for this result. {\bf Step 1:} Like the previous algorithm for \sssp on fully-connected network, we compute $n^{1/2}$ shortcuts from every node in $O(n^{1/2})$ time. Again, by \Cref{claim:spd reduction}, this gives us a graph $(G, w')$ such that $\spdiam(G, w')=O(n^{1/2})$. Additionally, every node sends these shortcuts to all other nodes (taking $O(n^{1/2})$ time).  
{\bf Step 2:} We then randomly pick  $n^{1/2}\poly\log n$ nodes and run the light-weight $h$-hop \sssp algorithm from these nodes, where $h=\spdiam(G, w')=O(n^{1/2})$. By \Cref{claim:k-source h-hop sp}, this takes $\tilde O(n^{1/2})$ time w.h.p. and gives us $(1+o(1))$-approximate values of the distances $\dist_{G, w}(x, v)$ between each random node $x$ and all other nodes $v$ (known by $v$). Each node $v$ broadcasts distances to these $n^{1/2}\poly\log n$ random nodes to all other nodes, taking $\tilde O(n^{1/2})$ time.

After this, we show that every node can use the information they have received so far to compute $(2+o(1))$-approximate values of its distances to all other nodes. (In particular, every node uses the distances it receives to build a graph and uses the distances in such graph as approximate distances between itself and other nodes.)\danupon{Does this also give a $2$-spanner or emulator? That would be weird!!!}
To explain the main idea, we show how to prove a $(3+o(1))$ approximation factor instead of $2+o(1)$: Consider any two nodes $u$ and $v$, and let $P$ be a shortest path between them. If $v$ is one of the $n^{1/2}$ nodes nearest to $u$, then $u$ already knows $\dist_{G, w}(u, v)$ from the first step (when we compute shortcuts). Otherwise, by a standard hitting set argument, one of these $n^{1/2}$ nearest nodes must be picked as one of  $n^{1/2}\poly\log n$ random nodes; let $x$ be such a node. Observe that $\dist_{G, w}(u, x)\leq \dist_{G, w}(u, v)$. By triangle inequality 
$$\dist_{G, w}(x, v)\leq \dist_{G, w}(x, u)+\dist_{G, w}(u, v) \leq 2\dist_{G, w}(u, v).$$ 
Again, by triangle inequality, 
$$\dist_{G, w}(u, v)\leq \dist_{G, w}(u, x)+\dist_{G, w}(x, v) \leq 3\dist_{G, w}(u, v);$$ 
in other words, $\dist_{G, w}(u, x)+\dist_{G, w}(x, v)$ is a $3$-approximate value of $\dist_{G, w}(u, v)$. Note that $u$ knows the exact value of $\dist_{G, w}(u, x)$ (from the first step) and the $(1+o(1))$-approximate value of $\dist_{G, w}(x, v)$ (from the second step). So, it can compute a $(1+o(1))$-approximate value of $\dist_{G, w}(u, x)+\dist_{G, w}(x, v)$  which is a $(3+o(1))$-approximate value of $\dist_{G, w}(u, v)$. Using the same argument, $v$ can also compute a $(3+o(1))$-approximate value of $\dist_{G, w}(u, v)$. To extend this idea to a $(2+o(1))$-approximation algorithm, we use exactly the same algorithm but has to consider a few more cases.

\paragraph{\sssp on General Networks (details in \Cref{sec:sssp})} Approximating \sssp in sublinear time needs both tools above and a few other ideas. First, we let $S$ be a set of $\frac{n^{1/2}}{\diam^{1/4}}\poly\log n$ random nodes and the source $s$. We need the following.
\begin{claim}[details in \Cref{sec:overlay}]\label{claim:overlay network is enough}
Let $h=n^{1/2}\diam^{1/4}$. Every node $v$ can compute an approximate distance to $s$ if it knows (i) approximate $h$-hop distances between itself and all nodes in $S$, and (ii) distances between the source $s$ and all nodes in $S$ in the following weighted graph $(G', w')$: nodes in $G'$ are those in $S$, and every edge $uv$ in $G'$ has weight equal to the $h$-hop distance between $u$ and $v$ in $G$.
\end{claim}
We call graph $(G', w')$ an overlay network since it can be viewed as a network sitting on the original network $(G, w)$. The idea of using the overlay network to compute distances is not new. It is a crucial tool in the context of dynamic data structures and distance oracle (e.g. \cite{DemetrescuFI05}). In distributed computing literature, it has appeared (in a slightly different form) in, e.g., \cite{LenzenP13}. 

Our main task is now to achieve (i) and (ii) in \Cref{claim:overlay network is enough}. Achieving (i) is in fact very easy: We simply run our light-weight $h$-hop \sssp from all nodes in $S$. By \Cref{claim:k-source h-hop sp}, this takes time $\tilde O(|S|+h)=\tilde O(n^{1/2}\diam^{1/4})$.\footnote{Note that nodes actually only know $(1+o(1))$ distances. To keep our discussion simple, we will pretend that they know the real distance.} In fact, by doing this we already partly achieve (ii): every node in $S$ already know the $h$-hop distance to all other nodes in $S$, thus it already has a ``local'' perspective in the overlay network $(G', w')$. To finish (ii), it is left to solve \sssp on $(G', w')$. 

To do this, we will first reduce the shortest-path diameter of the overlay network $(G', w')$ by creating $k$ shortcuts, where $k=\diam^{1/2}$. As noted in the \sssp algorithm on fully-connected network, it is enough for every node in $(G', w')$ to send $k$ lightest-weight edges incident to it to all other nodes (since running $k$ rounds of Dijkstra's algorithm will only need these edges). Broadcasting each such edge can be done in $O(\diam)$ time via the breadth-first search tree, and broadcasting all $k|S|=\tilde O(n^{1/2}\diam^{1/4})$ edges takes $\tilde O(n^{1/2}\diam^{1/4}+\diam)$ time by pipelining. (See details in \Cref{sec:reduce spd overlay}.) Let $(G'', w'')$ be an overlay network obtained from adding $k$ shortcuts to $(G', w')$. (As usual, nodes $u$ in $(G'', w'')$ only know weights $w''(uv)$ of edges $uv$ incident to it.) By \Cref{claim:spd reduction}, $\spdiam(G'', w'') = O(|S|/k) = \tilde O(n^{1/2}/\diam^{3/4})$.
Finally, we simulate our light-weight $h'$-hop \sssp algorithm to solve \sssp from source $s$ on overlay $(G'', w'')$, where $h'=\spdiam(G'', w'') = \tilde O(n^{1/2}/\diam^{3/4})$. To do this efficiently, we need a slightly stronger property of our light-weight $h'$-hop \sssp algorithm: recall that we have claimed that in our light-weight \sssp algorithm, each node sends a message through each edge only $\tilde O(1)$ times. In fact, we can show the following stronger claim.
\begin{claim}[details in \Cref{thm:light-weight sssp}]\label{claim:light-weight sssp broadcast}
In the light-weight \sssp algorithm, each node communicates in each round by {\em broadcasting} the same message to its neighbors. Moreover, each node broadcasts messages only for $\tilde O(1)$ times.
\end{claim}
The intuition behind the above claim is simple: at the heart of our light-weight \sssp algorithm, we solve $\tilde O(1)$ breadth-first search algorithms where, for each of these algorithms, each node broadcasts only once; it broadcasts its distance to the root, say $d$, at time $d$.
Now we simulate our light-weight \sssp algorithm on $(G'', w'')$ as follows. When each node $v$ wants to broadcast a message to all its neighbors in $G''$, we broadcast this message to all nodes in $G$, using the breadth-first search tree of $G$ (see details in \Cref{sec:sssp overlay}). 
This takes $O(\diam)$ time. If we want to broadcast $M_i$ messages in a round $i$ of our light-weight \sssp algorithm, we can do so in $O(\diam+M_i)$ time by pipelining. It can then be shown that the time we need to simulate all $r=\tilde O(h')=\tilde O(n^{1/2}/\diam^{3/4})$ rounds of our light-weight $h'$-hop \sssp algorithm takes $\tilde O(\diam h' +\sum_{i=1}^{r} M_i) = \tilde O(n^{1/2}\diam^{1/4} + \diam)$ (note that $\sum_{i=1}^{r} M_i = \tilde O(|S|)$ by \Cref{claim:light-weight sssp broadcast}). (See details in \Cref{sec:sssp time analysis}.) This completes (ii) in \Cref{claim:overlay network is enough}, and thus we can solve \sssp on $(G, w)$ in $\tilde O(n^{1/2}\diam^{1/4}+\diam)$ time.




\section{Main Tools} \label{sec:tools}

\subsection{Light-Weight Bounded-Hop Single-Source and Multi-Source Shortest Paths}\label{sec:light-weight bounded-hop sssp}

A key tool for our algorithm is a simple idea for computing a bounded-hop single-source shortest path and its extensions. Informally, an $h$-hop path is a path containing at most $h$ edges. The $h$-hop distance between two nodes $u$ and $v$ is the minimum weight among all $u$-$v$ $h$-hop paths. The problem of $h$-hop \sssp is to find the $h$-hop distance between a given source node $s$ and all other nodes. Formally:
\begin{definition}[$h$-hop \sssp]\label{def:bounded hop sssp}
Consider any network $G$ with edge weight $w$ and integer $h$. For any nodes $u$ and $v$, let $\cP^h(u, v)$ be a set of $u$-$v$ paths containing at most $h$ edges. We call $\cP^h(u, v)$ a set of $h$-hop $u$-$v$ paths. 
Define the $h$-hop distance between $u$ and $v$ as 
\begin{align*}
\dist_{G, w}^{h}(u, v) &= 
\begin{cases}
\min_{P\in \cP^h(u, v)} w(P) & \mbox{if $\cP^h(u, v)\neq \emptyset$}\\
\infty & \mbox{otherwise.}
\end{cases}
\end{align*}
%
\danupon{To define: $w(P)$.} 
Let $h$-hop \sssp be the problem where, for a given weighted network $(G, w)$, source node $s$ (node $s$ knows that it is the source), and integer $h$ (known to every node), we want every node $u$ to know $\dist_{G, w}^{h}(s, u)$.\danupon{Does $u$ always know that $\dist_{G, w}^{h}(s, u)=\infty$?}
\end{definition}
This problem can be solved in $O(h+D)$ time using, e.g., the distributed version of Bellman-Ford's algorithm. However, previous algorithms are ``heavy-weight'' in the sense that they require so much communication (i.e., there could be as large as $\Omega(h)$ messages sent through an edge) and thus are not suitable for parallelization. In this paper, we show a simple algorithm that can $(1+o(1))$-approximate $h$-hop \sssp in $\tilde O(h+D)$ time. Our algorithm is {\em light-weight} in the sense that every node broadcasts a message (of size $O(\poly\log n)$) to their neighbors only $O(\log n)$ times:
%
%
\begin{restatable}[Light-weight $h$-hop \sssp algorithm; proof in \Cref{sec:proof of them:light}]{theorem}{lighthsp}\label{thm:light-weight sssp}\label{thm:light}
There is an algorithm that solves $h$-hop \sssp on network $G$ with weight $w$\danupon{Don't forget that $w$ must be positive} in $\tilde O(h+D)$-time and, during the whole computation, every node $u$ broadcasts $O(\log n)$ messages, each of size $O(\poly\log n)$, to its neighbors $v$. 
\end{restatable}
\Cref{thm:light}, in its own form, cannot be directly used. We will extend it to an algorithm for computing $h$-hop {\em multi-source shortest paths} (\mssp). (Later, in \Cref{sec:overlay} we will also extend this result to {\em overlay networks}.)
%
%
The rest of this subsection is devoted to proving \Cref{thm:light-weight sssp}.\danupon{To do: Add some overview as commented below}




\subsubsection{Reducing Bounded-Hop Distance by Approximating Weights}


\danupon{To do: Explain. Below says that we can approximate bounded-hop distance $\dist^h_{G, w}$ by solving bounded-distance $\dist^{O(h)}_{G, w'_i}$ instead.}\danupon{To do: Recall what $d^h$ is.}

\begin{proposition}[Reducing Bounded-Hop Distance by Approximating Weights]\label{thm:approximate weight}
Consider any $n$-node weighted graph $(G, w)$ and an integer $h$. Let $\epsilon=1/\log n$. For any $i$ and edge $xy$, let $D'_i=2^i$ and $w'_i(xy)=\lceil \frac{2h w(xy)}{\epsilon D'_i}\rceil$. 
For any nodes $u$ and $v$, if we let 
%
\[
\widetilde \dist^h_{G, w}(u, v) = \min  \left\{\frac{\epsilon D'_{i}}{2h}\times \dist_{G, w'_i}(u, v) \mid i:\ \dist_{G, w'_i}(u, v)\leq (1+2/\epsilon)h\right\},
\]
%
%
%
then $\dist^h_{G, w}(u, v)\leq \widetilde \dist^h_{G,w}(u, v)\leq (1+\epsilon)\dist^h_{G, w}(u, v)$.
\end{proposition}

Note that the min term in \Cref{thm:approximate weight} is over all $i$ such that $\dist_{G, w'_i}(u, v)\leq (1+2/\epsilon)h$. The proof of \Cref{thm:approximate weight} heavily relies on \Cref{thm:approximate weight detail} below.

\begin{lemma}[Key Lemma for Reducing Bounded-Hop Distance by Approximating Weights]\label{thm:approximate weight detail}
Consider any nodes $u$ and $v$. For any $i$, let $w_i$ and $D'_i$ be as in  \Cref{thm:approximate weight}. Then, 
\begin{align}
\frac{\epsilon D'_i}{2h}\times \dist_{G, w'_i}(u, v)\geq \dist^h_{G, w}(u, v).\label{eq:apsp approx main one}
\end{align}
Moreover, for $i^*$  such that $D'_{i^*-1}\leq \dist^h_{G, w}(u, v)\leq D'_{i^*}$, we have that
\begin{align}
\dist_{G, w'_{i^*}}(u, v)&\leq (1+2/\epsilon)h ~~\mbox{and}\label{eq:apsp approx main two}\\
\frac{\epsilon D_{i^*}}{2h}\times \dist_{G, w'_{i^*}}(u, v) &\leq (1+\epsilon)\dist^h_{G, w}(u, v).\label{eq:apsp approx main three}
\end{align}
\end{lemma}
\begin{proof}
Let $P=\langle u=x_0, x_1, \ldots, x_\ell=v\rangle$ be any shortest $h$-hop path between $u$ and $v$ (thus $2^{i^*-1}\leq w(P)\leq 2^{i^*}$ and $\ell \leq h$). Then,
\begin{align}
\dist_{G, w'_{i^*}}(u, v)&=\sum_{j=0}^{\ell-1} \left\lceil \frac{2h w(x_jx_{j+1})}{\epsilon D'_{i^*}}\right\rceil\nonumber\\ 
&\leq \frac{2h}{\epsilon D'_{i^*}}\sum_{j=0}^{\ell-1} w(x_jx_{j+1}) + \ell\nonumber\\ 
&\leq \frac{2h}{\epsilon D'_{i^*}}\dist^h_{G, w}(u, v) + h \label{eq:apsp approx three}\\
&\leq (1+2/\epsilon)h \nonumber 
\end{align}
where the last inequality is because $D'_{i^*}\geq \dist^h_{G, w}(u,v)$. This proves \Cref{eq:apsp approx main two}.
Using \Cref{eq:apsp approx three}, we also have that
\begin{align*}
\frac{\epsilon D'_{i^*}}{2h}\times \dist_{G, w'_{i^*}}(u, v)  & \leq \frac{\epsilon D'_{i^*}}{2h}\times \left(\frac{2h}{\epsilon D'_{i^*}}\dist^h_{G, w}(u, v) + h\right)\\
&\leq \dist^h_{G, w}(u, v)+\frac{\epsilon}{2}D'_{i^*}\\
&\leq (1+\epsilon)\dist^h_{G, w}(u, v)
\end{align*}
where the last inequality is because $D'_{i^*}=2D'_{i^*-1}\leq 2\dist^h_{G, w}(u, v)$. This proves \Cref{eq:apsp approx main three}. Finally, observe that for any $i$ and the path $P$ defined as before, we have
\begin{align*}
\dist_{G, w'_{i}}(u, v)&\geq \sum_{j=0}^{\ell-1} \frac{2h w(x_jx_{j+1})}{\epsilon D'_i}\nonumber\\
&= \frac{2h}{\epsilon D'_i}  \dist^h_{G, w}(u, v)\,.
\end{align*}
It follows that
\begin{align*}
\frac{\epsilon D'_i}{2h}\times\dist_{G, w'_i}(u, v) & \geq \frac{\epsilon D'_i}{2h}\times \left(\frac{2h}{\epsilon D'_i}  \dist^h_{G, w}(u, v)\right)\\
&= \dist^h_{G, w}(u, v)\,.
\end{align*}
This proves \Cref{eq:apsp approx main one} and completes the proof of \Cref{thm:approximate weight detail}.
\end{proof}

Now we are ready to prove \Cref{thm:approximate weight}.

\begin{proof}[Proof of \Cref{thm:approximate weight}] Note that
\begin{align*}
\tilde d^h_{G, w}(u, v) &= \min  \left\{\frac{\epsilon D'_{i}}{2h}\times \dist_{G, w'_i}(u, v) \mid i:\ \dist_{G, w'_i}(u, v)\leq (1+2/\epsilon)h\right\}\\
&\leq  \frac{\epsilon D'_{i^*}}{2h}\times \dist_{G, w'_{i^*}}(u, v)\\
&\leq (1+\epsilon)\dist^h_{G, w}(u, v)
\end{align*}
%
%
where $i^*$ is as in \Cref{thm:approximate weight detail}, the second inequality is due to the fact that $\dist_{G, w'_{i^*}}(u, v)\leq (1+2/\epsilon)h$ as in \Cref{eq:apsp approx main two}, and the third inequality follows from \Cref{eq:apsp approx main three}. This proves the second inequality in \Cref{thm:approximate weight}. The first inequality of \Cref{thm:approximate weight} simply follows from the fact that $\frac{\epsilon D'_i}{2h}\times \dist_{G, w'_i}(u, v)\geq \dist^h_{G, w}(u, v)$ for all $i$, by \Cref{eq:apsp approx main one}. 
%
\end{proof}

%
%
%

\subsubsection{Algorithm for Bounded-Hop \sssp (Proof of \Cref{thm:light})}\label{sec:proof of them:light}

We now show that we can solve $h$-hop \sssp in $\tilde O(h+D)$ time while each node broadcasts $O(\log n)$ messages of size $O(\log n)$, as claimed in \Cref{thm:light}. Our algorithm is outlined in \Cref{algo:bounded-hop sssp}. Given a parameter $h$ (known to all nodes) and weighted network $(G, w)$, it computes $w_i$, for all $i$, as defined in \Cref{thm:approximate weight}; i.e., every node $u$ internally computes $w_i(u, v)$ for all neighbors $v$. Note that this step needs no communication. Next, in Line~\ref{line:solve bounded distance sssp} of \Cref{algo:bounded-hop sssp}, for each value of $i$, the algorithm executes an algorithm for the bounded-distance \sssp problem with parameter $(G, w, s, K)$, where $(1+2/\epsilon)h$, as outlined in \Cref{algo:pseudo polynomial} (we will explain this algorithm next). At the end of the execution of \Cref{algo:pseudo polynomial}, every node $u$ knows $\dist_i'(s, u)$ such that $\dist_i'(s, u)=\dist_{G, w_i}(s, u)$ if $\dist_{G, w_i}(s, u)\leq K$ and $\dist_i'(s, u)=\infty$ otherwise. Finally, we set $\widetilde \dist^h_{G, w_i}(s, u) = \min_i \dist'_i(s, u)$. By \Cref{thm:approximate weight}, we have that $\dist^h_{G, w}(u, v)\leq \widetilde \dist^h_{G,w}(u, v)\leq (1+\epsilon)\dist^h_{G, w}(u, v) = (1+o(1))\dist^h_{G, w}(u, v)$ as desired.\danupon{Careful!!! If we use $\epsilon=1/\log n$, how many different values of $(1+\epsilon)^i$ are there? } 

We now explain \Cref{algo:pseudo polynomial} for solving the bounded-distance \sssp problem. It is a simple modification of a standard bread-first tree algorithm. It runs for $K$ rounds. In the initial round (Round $0$), the source node $s$ broadcasts a message $(s, 0)$ to all its neighbors to start the algorithm. This message is to inform all its neighbors that its distance from the source (itself) is $0$. In general, we will make sure that every node $v$ whose distance to $s$ is $\dist_{G, w}(s, u)=\ell$ will broadcast a message $(s, \ell)$ to its neighbor at Round $\ell$. Every time a node $u$ receives a message of the form $(s, \ell)$ from its neighbor $v$, it knows that $\dist_{G, w}(s, v)=\ell$; so, $u$ updates its distance to the minimum between the current distance and $\ell+w(uv)$. 
%
It is easy to check that every node $u$ such that $\dist_{G, w}(s, u)< K$ broadcasts its message to all neighbors {\em once} at Round $\ell=\dist_{G, w}(s, u)$. 
%
The correctness of \Cref{algo:pseudo polynomial} immediately follows. Moreover, since we execute \Cref{algo:pseudo polynomial} for $O(\log n)$ different values of $i$ (since the maximum weight is $\poly(n)$), it follows that every node broadcasts a message to their neighbors $O(\log n)$ times. \Cref{thm:light} follows.


\begin{algorithm}
\caption{Bounded-Hop \sssp $(G, w, s, h)$} \label{algo:bounded-hop sssp}
{\bf Input:} Weighted undirected graph $(G, w)$, source node $s$, and integer $h$. 

{\bf Output:} Every node $u$ knows the value of $\widetilde \dist^h_{G, w_i}(s, u)$ such that $\dist^h_{G, w}(s, u)\leq \widetilde \dist^h_{G, w_i}(s, u) \leq (1+o(1))\dist^h_{G, w_i}(s, u)$.





\begin{algorithmic}[1]


\State Let $\epsilon=1/\log n$. For any $i$ and edge $xy$, let $D'_i=2^i$ and $w'_i(xy)=\lceil \frac{2h w(xy)}{\epsilon D'_i}\rceil$. Let $K=(1+2/\epsilon)h$.

\State Let $t$ be the time this algorithm starts. We can assume that all nodes know $t$.

\ForAll{i}

%
%
%
%
%
%
%
%
%
%
%
%
%
%
%

\State Solve bounded-distance \sssp with parameters $(G, w_i, s, K)$
using \Cref{algo:pseudo polynomial}. 
(This takes  $\tilde O(K)=\tilde O(h)$ time.)
Let $\dist'_i(s, u)$ be the distance returned to node $u$.\label{line:solve bounded distance sssp}

\EndFor

\State Each node $u$ computes $\widetilde \dist^h_{G, w_i}(s, u) = \min_i \dist'_i(s, u)$.

\end{algorithmic}
\end{algorithm}

\begin{algorithm}
\caption{Bounded-Distance \sssp $(G, w, s, K)$} \label{algo:pseudo polynomial}
{\bf Input:} Weighted undirected graph $(G, w)$, source node $s$, and integer $K$. 

{\bf Output:} Every node $u$ knows $\dist'_{G, w}(s, u)$ where $\dist'_{G, w}=\dist_{G, w}(s, u)$ if $\dist_{G, w}(s, u)\leq K$ and $\dist'_{G, w}=\infty$ otherwise. 


%
%
%
%

\begin{algorithmic}[1]


\State Let $t$ be the time this algorithm starts. We can assume that all nodes know $t$.

\State Initially, every node $u$ sets $\dist'_{G, w}(s, u)=\infty$. 

\State In the beginning of this algorithm (i.e., at time $t$) source node $s$ sends a message $(s, 0)$ to itself.



\If{a node $u$ receives a message $(s, \ell)$ for some $\ell$ from node $v$,}


\If{($\ell+w(u, v)\leq K$) and ($\ell+w(u, v)<\dist'_{G}(s, u)$)}

\State $u$ sets $\dist'_{G}(s, u)=\ell+w(u, v)$.


\EndIf
\EndIf

\State For any $x\leq K$, at time $t+x$, every node $u$ such that $\dist'_{G}(s, u)=x$ broadcasts message $(s, x)$ to all its neighbors to announce that $\dist'_{G}(s, u)=x$.  

\end{algorithmic}
\end{algorithm}


\subsubsection{Bounded-Hop Multi-Source Shortest Paths} \label{sec:bounded-hop mssp}
The fact that our algorithm for the bounded-hop single-source shortest path problem in \Cref{thm:light-weight sssp} is light-weight allows us to solve its {\em multi-source} version, where there are {\em many sources} in parallel. The problem of bounded-hop multi-source shortest path is as follows. 
\begin{definition}[$h$-hop $k$-source shortest paths]\label{def:bounded hop ksp}
Given a weighted network $(G, w)$, integer $h$ (known to every node), and sources $s_1, \ldots, s_k$ (each node $s_i$ knows that it is a source), the goal of the $h$-hop $k$-source shortest paths problem is to make every node $u$ knows $\dist_{G, w}^{h}(s_i, u)$ for all $1\leq i\leq k$.
\end{definition}
The main result of this section is an algorithm for solving this problem in $\tilde O(k+h+D)$ time, as follows. 
\begin{restatable}[$k$-source $h$-hop shortest path algorithm]{proposition}{khsp}\label{thm:k-source h-hop sp}
There is an algorithm that $(1+o(1))$-approximates the $h$-hop $k$-source  shortest paths problem on weighted network $(G, w)$ in $\tilde O(k+h+D)$ time; i.e., at its termination every node $u$ knows $\dist'_{G, w}(s_i, u)$ such that 
\[
\dist_{G, w}^{h}(s_i, u)\leq \dist'_{G, w}(s_i, u)\leq (1+o(1))\dist_{G, w}^{h}(s_i, u)
\]
for all sources $s_i$. 
\end{restatable}
The algorithm is conceptually easy: we simply run the algorithm for bounded-hop single-source shortest path in \Cref{thm:light-weight sssp} (i.e. \Cref{algo:bounded-hop sssp}) from $k$ sources {\em in parallel}. Obviously, this algorithm needs at least $\tilde \Omega(h+D)$ time since this is the guarantee we can get for the case of single source. Moreover, it is possible that one need has to broadcast $O(\log n)$ messages for each execution of \Cref{algo:bounded-hop sssp}, making a total of $O(k\log n)$ messages; this will require $\tilde O(k)$ time. So, the best running time we can hope for is $\tilde O(k+h+D)$. 
%
%
It is, however, not obvious to achieve this running time since one execution could {\em delay} other executions; i.e., it is possible that all executions of \Cref{algo:bounded-hop sssp} might want a node to send a message at the same time making some of them unable to proceed to the next round. We show that by simply adding a small {\em delay} to the starting time of each execution, it is {\em unlikely} that many executions will delay each other. 
%

%
%
The algorithm is very simple: Instead of starting the execution of \Cref{algo:bounded-hop sssp} from different source nodes at the same time, each execution starts with a {\em random delay} randomly selected from integers from $0$ to $k\log n$. The algorithm is outlined in \Cref{algo:random delay}. 
\begin{algorithm}
\caption{Multi-Source Bounded-Hop Shortest Path $(G, w, \{s_1, \ldots, s_k\}, h)$} \label{algo:random delay}\label{algo:bounded-hop multi-source}
{\bf Input:} Weighted undirected graph $(G, w)$, $k$ source nodes $s_1, \ldots, s_k$, and integer $h$. 

{\bf Output:} Every node $u$ knows $\dist^h_{G, w}(s_i, u)$ for all $i$. 


%
%
%
%

\begin{algorithmic}[1]


\State Let $r_1, \ldots, r_k$ be a number selected uniformly at random from $[0, k\log n]$. We can assume that all nodes know $r_i$, for all $r_i$. This can be done by, e.g., broadcasting all $r_i$ to all nodes in $O(k+D)$ time. 

\State Let $t$ be the time this algorithm starts. We can assume that all nodes know $t$.

\State At time $t+r_i$, execute the bounded-hop single-source shortest path algorithm (\Cref{algo:bounded-hop sssp}) on $(G, w, s_i, h)$.

\State If at any time, more than $\log n$ messages is sent through an edge, we say that the algorithm {\em fails}. (We show that the algorithm fails with probability $O(1/n^2)$ in \Cref{thm:no congestion}.)

\end{algorithmic}
\end{algorithm}
The crucial thing is to show that many executions of \Cref{algo:bounded-hop sssp} launched by \Cref{algo:random delay} {\em do not delay each other}. In particular, that we show that at most $\log n$ messages will be sent through each edge in every round, with high probability (if this does not happen, we say that the algorithm {\em fails}). We prove this in \Cref{thm:no congestion} below. Our proof is simply an adaptation of the random delay technique for package scheduling \cite{LeightonMR94}. \Cref{thm:no congestion} immediately implies \Cref{thm:k-source h-hop sp}, since each execution, which start at time $\tilde O(k)$ will finish in $\tilde O(h+D)$ rounds without being delayed.   
\begin{lemma}[Congestion guaranteed by the random delay technique]\label{thm:no congestion}
For any source $s_i$ and node $u$, let $\mset_{i, u}$ be the set of messages broadcasted by $u$ during the execution of \Cref{algo:bounded-hop sssp} with parameter $(G, w, s_i, h)$. Note that $|\mset_{i, u}|\leq c\log n$ for some constant $c$, by \Cref{thm:light-weight sssp}. 
%
%
%
Then, the probability that, in \Cref{algo:random delay}, there exists time $t$, node $u$, and a set $\mset\subseteq \bigcup_i \mset_{i, u}$ such that $|\mset|\geq \log n$, and all messages in $\mset$ are broadcasted by $u$ at time $t=O(k+h+D)$, is $O(1/n^2)$. 
%
\end{lemma}
\begin{proof}
Fix any node $u$, time $t$, and set $\mset$ as above. Observe that, for any $i$, the time that a message $M\in \mset_{i, u}$, is broadcasted by $u$ is determined by the random delay $r_i$ -- there is only one value of $r_i$ that makes $u$ broadcasts $M$ at time $t$. In other words, for fixed $u$, $t$, and message $M\in \bigcup_i \mset_{i, u}$, 
\[
Pr[\mbox{$M$ is sent by $u$ at time $t$}]\leq \frac{1}{k\log n}. 
\] 
It follows that for fixed $u$, $t$, and set of messages $\mset$,
\[
Pr[\mbox{all messages in $\mset$ is sent by $u$ at time $t$}]\leq \left(\frac{1}{k\log n}\right)^{|\mset|}. 
\]
Note that we can assume that $|\mset\cap \mset_{i, u}|\leq 1$ since, for an execution of \Cref{algo:bounded-hop sssp} on a source $s_i$, every node $u$ broadcasts at most one message per round. This implies that $|\mset|\leq k$, and, for any $m\leq k$, the number of such set $\mset$ of size exactly $m$ is at most $\binom{k}{m}(c\log n)^m$ since each set $\mset$ can be constructed by picking $m$ different sets $\mset_{i, u}$, and picking one message out of $c\log n$ messages from each $\mset_{i, u}$. Thus, for fixed $u$ and $t$, the probability that there exists $\mset$ such that $|\mset|\geq \log n$ and all messages in $\mset$ is sent by $u$ at time $t$ is at most
\begin{align*}
\sum_{m=\log n}^{k} \binom{k}{m}(c\log n)^m\left(\frac{1}{k\log n}\right)^m. 
\end{align*}
Using the fact that for any $0<b<a$, $\binom{a}{b}\leq (ae/b)^b$, the previous quantity is at most
\begin{align*}
\sum_{m=\log n}^{k} \left(\frac{ke}{m}\right)^m (c\log n)^m\left(\frac{1}{k\log n}\right)^m
&\leq \sum_{m=\log n}^{k} \left(\frac{ec}{m}\right)^m\\
&\leq k\left(\frac{ec}{\log n}\right)^{\log n}.
\end{align*}
For large enough $n$, the above quantity is at most $1/n^4$. We conclude that for fixed $u$ and $t$, the probability that there exists $\mset$ such that $|\mset|\geq \log n$ and all messages in $\mset$ is sent by $u$ at time $t$ is at most  $1/n^4$. By summing this probability over all nodes $u$ and $t=O(k+h+D)=O(n)$, \Cref{thm:no congestion} follows.
\end{proof}

%
%

\subsection{Shortest-Path Diameter Reduction Using Shortcuts}\label{sec:spd reduction}

%
In this section, we show a simple way to augment a graph with some edges (called ``shortcuts'') to reduce the shortest-path diameter. The notion of shortest path diameter is defined as follows.  
\begin{definition}[Shortest-path distance and diameter]\label{def:spdiam}
For any weighted graph $(G, w)$, the {\em shortest-path distance} between any two nodes $u$ and $v$, denoted by $\spdist_{G, w}(u, v)$, is the minimum integer $h$ such that $\dist_{G, w}^h(u, v)=\dist_{G, w}(u, v)$. That is, it is the minimum number of edges among the shortest $u$-$v$ paths. The {shortest-path diameter} of $(G, w)$, denoted by $\spdiam(G, w)$, is defined to be $\max_{u, v\in V(G)} \spdist_{G, w}(u, v)$. In other words, it is the minimum integer $h$ such that $\dist_{G, w}^h(u, v)=\dist_{G, w}(u, v)$ for all nodes $u$ and $v$. 
\end{definition}
\begin{definition}[$k$-shortcut graph]\label{def:shortcut graph}
Consider any $n$-node weighted graph $(G, w)$ and an integer $k\neq n-1$. For any node $u$, let $S^k_{G, w}(u)\subseteq V(G)$ be the set of exactly $k$ nodes nearest to $u$ (excluding $u$); i.e. $u\notin S^k_{G, w}(u)$, $|S^k_{G, w}(u)|=k$, and for all $v\in S^k_{G, w}(u)$ and $v'\notin S^k_{G, w}(u)$, $\dist_{G, w}(u,v)\leq \dist_{G, w}(u, v')$. 
The {\em $k$-shortcut graph of $(G, w)$}, denoted by $(G, w)^k$, is a weighted graph resulting from adding an edge $uv$ of weight $\dist_{G, w}(u, v)$ for every $u\in V(G)$ and $v\in S^k_{G, w}(u)$. When it is clear from the context, we will write $S^k(u)$ instead of $S^k_{G, w}(u)$.
\end{definition}
%
%
\begin{proposition}[Main result of \Cref{sec:spd reduction}: Reducing the shortest-path diameter by shortcuts]\label{thm:shortcuts}\label{thm:spd reduction}
For any $n$-node weighted undirected graph $(G, w)$ and integer $k$, 
if $(G', w')$ is the $k$-shortcut graph of $(G, w)$, then $\spdiam(G', w')<  4n/k$. 
\end{proposition}
\begin{proof}
Consider any nodes $u$ and $v$, and let 
\[
P=\langle u=x_0, x_1, \ldots, x_\ell=v\rangle
\]
be the shortest $u$-$v$ path in $(G', w')$ with smallest number of edges; i.e. there is no path $Q$ in $(G', w')$ such that $|E(Q)|<|E(P)|$\danupon{Define $E(P)$ and $w'(P)$} and $w'(Q)\leq w'(P)$. 
%
%
For any node $x$, let $S^k(x)$ be the set of $k$ nodes nearest to $x$ in $(G, w)$, as in \Cref{def:shortcut graph}. We claim that for any integer $0\leq i\leq (\ell/4)-1$, we have 
\[S^k(x_{4i})\cap S^k(x_{4(i+1)})=\emptyset.\]  
This claim immediately implies that $\ell< 4n/k$, thus \Cref{thm:shortcuts}; otherwise, $|\bigcup_{0\leq i\leq \ell/4}S_k(x_{4i})| \geq (k+1)(n/k) > n$, which is impossible. It is thus left to prove the claim that $S^k(x_{4i})\cap S^k(x_{4(i+1)})=\emptyset$. 
Now, consider any $1\leq i\leq \ell/4$. Observe that 
\begin{align}
\label{eq:spd reduction one}
x_{4i+2}\notin S^k(x_{4i}).
\end{align}
Otherwise, $(G', w')$ will contain edge $x_{4i}x_{4i+2}$ of weight $\dist_{G, w}(x_{4i},x_{4i+2})$. This implies that $P'=\langle x_0, \ldots, x_{4i-1}, x_{4i}, x_{4i+2}, x_{4i+3}, \ldots, x_\ell=v\rangle$ is a shortest $u$-$v$ path in $(G', w')$ containing $\ell-1$ edges, contradicting the fact that $P$ has the smallest number of edges among shortest $u$-$v$ paths in $(G', w')$.  By the same argument, we have 
\begin{align}
\label{eq:spd reduction two}
x_{4i+2}\notin S^k(x_{4(i+1)}).
\end{align}
By the definition of $S^k$, \Cref{eq:spd reduction one,eq:spd reduction two} imply that 
\begin{align}
\label{eq:spd reduction three}\forall y\in S^k(x_{4i}) && \dist_{G, w}(x_{4i}, y) &\leq \dist_{G, w}(x_{4i}, x_{4i+2}), ~~~\mbox{and}\\
\label{eq:spd reduction four}\forall y\in S^k(x_{4(i+1)}) && \dist_{G, w}(x_{4(i+1)}, y) &\leq \dist_{G, w}(x_{4(i+1)}, x_{4i+2})
\end{align} 
respectively. Now, assume for a contradiction that there is a node $y\in S^k(x_{4i}) \cap  S^k(x_{4(i+1)})$. Consider a path 
\[
P''=\langle x_0, \ldots, x_{4i-1}, x_{4i}, y, x_{4(i+1)}, x_{4(i+1)+1}, \ldots, x_\ell=v\rangle.
\] 
Observe that $P''$ contains $\ell-3$ edges. Moreover, \Cref{eq:spd reduction three,eq:spd reduction four} imply that 
\begin{align*}
w'(x_{4i}y) &=\dist_{G, w}(x_{4i}, y)  \leq \dist_{G, w}(x_{4i}, x_{4i+2}) ~~~\mbox{and}~~~\\ 
w'(x_{4(i+1)}y) &= \dist_{G, w}(x_{4(i+1)}, y) \leq \dist_{G, w}(x_{4(i+1)}, x_{4i+2})
\end{align*}
which further imply that $w'(P'')\leq w'(P)$. This means that $P''$ is a shortest $u$-$v$ paths in $(G', w')$ and contradicts the fact that $P$ has the smallest number of edges among shortest $u$-$v$ paths in $(G', w')$. Thus,  $S^k(x_{4i}) \cap  S^k(x_{4(i+1)})=\emptyset$ as desired. 
\end{proof}

We note a simple fact that will be used throughout this paper: we can compute $S^k_{G, w}(u)$ and $\dist_{G, w}(u, v)$ for all $v\in S^k_{G, w}(u)$ if we know $k$ smallest edges incident to {\em every} nodes. The precise statement is as follows.


\begin{definition}[$E^k(u)$ and $(G^k, w)$]\label{def:E^k(u) set of k smallest weight }
For any node $u$, let $E^k(u)$ be the set of $k$ edges incident to $u$ with minimum weight (breaking tie arbitrarily); i.e., for every edge $uv\in E^k(u)$ and $uv'\notin E^k(u)$, we have $w(uv)\leq w(uv')$. Let $(G^k, w)$ be the subgraph of $(G, w)$ whose edge set is $\bigcup_{v\in V(G)} E^k(v)$.
\end{definition}

We note that for some graph $(G, w)$, the sets $E^k(u)$ and $S^k_{G, w}(u)$ might not be uniquely defined. To simplify our statement and proofs, we will assume that $(G, w)$ has the following property, which makes both $S^k_{G, w}(u)$ and $S^k_{G^k, w}(u)$ unique: every edge $uv\in E(G)$ has a unique value of $w(uv)$, and every pair of nodes $u$ and $v$ has a unique value of $\dist_{G, w}(u, v)$ and $\dist_{G^k, w}(u, v).$ Removing these assumptions can be easily done by breaking ties arbitrarily. 

\begin{observation}[Computing shortcut edges using $(G^k, w)$]\label{thm:shortcuts from smallest edges}
For any node $u$, 
\begin{itemize}
\item $S^k_{G, w}(u)=S^k_{G^k, w}(u)$, and
\item $\dist_{G, w}(u, v)= \dist_{G^k, w}(u, v)$, for any $v\in S^k_{G, w}(u)$.
\end{itemize}
In other words, using only edges in  $\bigcup_{v\in V(G)} E^k(v)$ and their weights, we can compute all $k$-shortcut edges.
\end{observation}
\begin{proof} The intuition behind \Cref{thm:shortcuts from smallest edges} is that Dijkstra's algorithm can be used to compute  $S^k_{G, w}(u)$ and $\{\dist_{G, w}(u, v)\}_{v\in S^k_{G, w}(u)}$ by executing it for $k$ iterations, and this process will never need any edge besides those in $(G^k, w)$. Below we provide a formal proof that does not require the knowledge of Dijkstra's algorithm. 

Let $T_u$ be a shortest path tree rooted at $u$ in $(G, w)$. Assume for a contradiction that there is a node $v\in  S^k_{G, w}(u)\setminus  S^k_{G^k, w}(u)$. Let $z$ be the parent of $v$ in $T_u$. The fact that $v\notin S^k_{G^k, w}(u)$ implies that $zv\notin E^k(z)$\danupon{More detail?}; thus, there exists $v'$ in $E^k(z)\setminus S^k_{G^k, w}(u)$ (since $|E^k(z)|= |S^k_{G^k, w}(u)|=k$ and $S^k_{G^k, w}(u)\setminus E^k(z)\neq \emptyset$). 
Note that since $v\notin E^k(z)$, we have $w(zv')<w(zv)$ (recall that we assume that edge weights are distinct). This, however, implies that 
\[
\dist_{G, w}(u, v')\leq \dist_{G, w}(u, z)+w(zv') < \dist_{G, w}(u, v).
\]
This contradicts the fact that $v\in S^k_{G, w}(u)$ and  $v\notin S^k_{G, w}(u)$.
\end{proof}

\section{Algorithms on General Networks}
\label{sec:sssp}
\danupon{The time vs approx in the title is not the best we can say.}

%
%
%
%
%
%
%
%
%
%
%
%
%
%
%
%
%
%
%

In this section, we present algorithms for \sssp and \apsp on general distributed networks, as stated in \Cref{thm:main sssp} and \Cref{thm:linear apsp}. First, observe that \apsp is simply a special case of the $h$-hop $k$-source shortest paths problem defined in \Cref{sec:bounded-hop mssp} where we use $h=k=n$. Thus, by \Cref{thm:k-source h-hop sp}, there is an $(1+o(1))$-approximation algorithm that solves \apsp in $\tilde O(k+h+D)=\tilde O(n)$ time with high probability. This immediately proves \Cref{thm:linear apsp}. 
%
%
The rest of this section is then devoted to showing a $\tilde O(n^{1/2}\diam^{1/4}+\diam)$-time algorithm for \sssp as in \Cref{thm:main sssp}, which require several non-trivial steps.


\subsection{Reduction to Single-Source Shortest Path on Overlay Networks}\label{sec:skeleton}\label{sec:overlay}

In this section, we show that solving the single source shortest path problem on a network $(G, w)$ can be reduced to the same problem on a certain type of an {\em overlay network}, usually known as a {\em landmark} or {\em skeleton} (e.g. \cite{Sommer12-survey,LenzenP13}). 
In general, an overlay network $G'$ is a {\em virtual} network of nodes and logical links that is built on top of an underlying {\em real} network $G$; i.e., $V(G')\subset V(G)$ and an edge in $G'$ (a ``virtual edge'') corresponds to a path in $G$ (see, e.g., \cite{EmekFKKP12}). Its implementation is usually abstracted as {\em a routing scheme} that maps virtual edges to underlying routes. However, for the purpose of this paper, we do not need a routing scheme but will need the notion of {\em hop-stretch} which captures the number of hops in $G$ between two neighboring virtual nodes in $V(G')$, as follows.
%
%
%
%
%
\begin{definition}[Overlay network of hop-stretch $\lambda$]\label{def:overlay}
Consider any network $G$. For any $\lambda$, a weighted network $(G', w')$ is said to be an {\em overlay network of hop-stretch $\lambda$ embedded in $G$} if 
\begin{enumerate}
\item $V(G')\subseteq V(G)$,
\item $\dist_G(u, v)\leq \lambda$ for every virtual edge $uv\in E(G')$, and
\item for every virtual edge $uv\in E(G')$, both $u$ and $v$ (as a node in $G$) knows the value of $w'(uv)$.
\end{enumerate}  
\end{definition}

We emphasize that $\lambda$ captures the number of hops ($\dist_G(u, v)$) between two neighboring nodes $u$ and $v$, not the weighted distance ($\dist_{G, w}(u, v)$). The main result of this section is an algorithm to construct an overlay network such that, if we can solve the single-source shortest path problem on such network, we can solve the single-source shortest path on the whole graph: 


\begin{proposition}[Main result of \Cref{sec:skeleton}: Reduction to an overlay network]\label{thm:skeleton}\label{thm:overlay network}
For any weighted graph $(G, w)$, source node $s$, and integer $\alpha$, there is an $\tilde O(\alpha+n/\alpha+D)$-time distributed algorithm that embeds an overlay network $(G', w')$ in $G$ such that, with high probability, 
\begin{enumerate} 
\item $s\in V(G')$,
\item $|V(G')|=\tilde O(\alpha)$, and
\item if every node $u\in V(G)$ knows a $(1+o(1))$-approximate value of $\dist_{G', w'}(s, v)$ for every node $v\in V(G')$, then $u$ knows the  $(1+o(1))$-approximate value of $\dist_{G, w}(s, u)$. 
\end{enumerate}
\end{proposition}
\begin{proof}
Our algorithm is as follows. First, every node $u$ selects itself to be in $V(G')$ with probability $\alpha/n$. Additionally, we always keep source $s$ in $V(G')$; this guarantees the first condition. Observe that $E[|V(G')|] = 1+(\alpha/n)\cdot (n-1) \leq 2\alpha$; so, by Chernoff's bound (e.g. \cite[Theorem 4.4]{MitzenmacherU05}), 
$Pr[|(G')|\geq 12\alpha\log n] \leq 1/n.$
This proves the second condition. 
%
To guarantee the last condition, we have to define edges and their weights in $G'$. To do this, we invoke an algorithm for the bounded-hop multi-source shortest path problem in \Cref{thm:k-source h-hop sp} (page \pageref{thm:k-source h-hop sp}), with nodes in $V(G')$ as sources and $h=n\log n/\alpha$ hops. By \Cref{thm:k-source h-hop sp}, the algorithm takes $\tilde O(|V(G')|+h+D)=\tilde O(\alpha+n/\alpha+D)$ time, and every node $u\in V(G)$ will know $\dist'_{G, w}(u, v)$ such that 
\[
\dist_{G, w}^{h}(u, v)\leq \dist'_{G, w}(u, v)\leq (1+o(1))\dist_{G, w}^{h}(u, v)
\]
for all $v\in V(G')$. For any $u, v\in V(G')$ such that $\dist'_{G, w}(u, v)<\infty$, we add edge $uv$ with weight $w'(u, v)=\dist'_{G, w}(u, v)$ to $G'$. Note that both $u$ and $v$ knows the existence and weight of this edge. This completes the description of an overlay network $(G', w')$ embedded in $G$. 

We are now ready to show the third condition, i.e., if a node $u\in V(G)$ knows a $(1+o(1))$-approximate value of $\dist_{G', w'}(s, v)$ for all $v\in V(G')$, then it knows a $(1+o(1))$-approximate value of $\dist_{G, w}(s, u)$. Consider any node $u\in V(G)$, and let 
$$P=\langle u=v_0, v_1, \ldots, v_k=s\rangle,$$ 
for some $k$, be a shortest path between $s$ and $u$ in $(G, w)$. Observe that if $k\leq n\log n/\alpha$, then $u$ knows $\dist'_{G, w}(u, v)$ which is a $(1+o(1))$-approximate value of $\dist_{G, w}(s, u)$, and thus the third condition holds even when $u$ does not know a $(1+o(1))$-approximate value of $\dist_{G', w'}(s, v)$ for any $v\in V(G')$. It is thus left to consider the case where $k\geq n\log n/\alpha$. 
Let $i_1 < i_2 < \ldots < i_t$ be such that $v_{i_1}, \ldots, v_{i_t}$ are nodes in $P\cap V(G')$. Let $i_0=0$ (i.e., $v_{i_0}=u$). We note the following simple fact, which is very easy and well-known (e.g. \cite{UllmanY91}). We provide its proof here only for completeness. 
\begin{lemma}[Bound on the number of hops between two landmarks in a path]\label{thm:landmark}
For any $j$, $i_j-i_{j-1}\leq n \log n/\alpha$, with probability at least $1-2^{-\beta n}$, for some constant $\beta>0$ and sufficiently large $n$.
\end{lemma}
\begin{proof}
We note a well-known fact that a set of random selected nodes $|V(G')|$ of size $\alpha$ will ``hit'' a simple path of length at least $cn\log n/|V(G')$, for some constant $c$, with high probability. To the best of our knowledge, this fact was first shown in \cite{GreenKnuth81book} and has been used many times in dynamic graph algorithms (e.g. \cite{DemetrescuFI05} and references there in). The following fact appears as Theorem 36.5 in \cite{DemetrescuFI05} (attributed to \cite{UllmanY91}). 
%
%
%
\begin{fact}[Ullman and Yannakakis \cite{UllmanY91}]\label{thm:UllmanY} 
Let $S \subseteq V(G)$ be a set of vertices chosen uniformly at random. Then the probability that a given simple path has a sequence of more than $(cn\log n)/|S|$ vertices, none of which are from $S$, for any $c > 0$, is, for sufficiently large $n$, bounded by $2^{-\beta n}$ for some positive $\beta$.
\end{fact}
Using $S=V(G')$, which has size $\tilde \Theta(\alpha)$, we have that every subpath of $P$ of length at least $n\log n/\alpha$ contains a node in $V(G')$, with high probability. The lemma follows by union bouding over the subpaths of $P$.\danupon{I'm a bit sketchy here.} This completes the proof of \Cref{thm:landmark}.
\end{proof}
It follows from the above lemma that $u$ knows $\dist'_{G, w}(u, v_{i_1})$ and, for any $j\geq 1$, $v_{i_j}v_{i_{j+1}}$ is an edge in the overlay network $G'$ of weight $w'(v_{i_j}v_{i_{j+1}})=\dist'_{G, w}(v_{i_j}, v_{i_{j+1}})$, with probability at least $1-2^{-\beta n}$. Thus, with high probability, 
\[\dist_{G', w'}(v_{i_1}, s)\leq \sum_{j=1}^{k-1} \dist'_{G, w}(v_{i_j}, v_{i_{j+1}}).\] 
Since $u$ already knows $\dist'_{G, w}(u, v_{i_1})$, it can now compute  
\[\dist'_{G, w}(u, v_{i_1})+\sum_{j=1}^{k-1} \dist'_{G, w}(v_{i_j}, v_{i_{j+1}})\] 
which is at least $\dist_{G, w}(u, s)$ and at most $(1+o(1))\dist_{G, w}(u, s)$. We note one detail that, in fact, $u$ does not known which node is $v_1$, so it has to use the value of
\[\min_{v\in V(G')} \dist'_{G, w}(u, v) + \dist_{G', w'}(v, s)\] 
as an estimate.\danupon{Vague} By union bounding over all nodes $u$, \Cref{thm:overlay network} follows.\danupon{I'm quite sketchy here.}
\end{proof}

\subsection{Reducing the Shortest Path Diameter of Overlay Network $(G', w')$}\label{sec:reduce spd overlay}
In this section, we assume that we are given an overlay network $(G', w')$ embedded in the original network $(G, w)$, as show in \Cref{thm:overlay network}. Our goal is to solve the single-source shortest path problem on $(G', w')$. Recall that $(G', w')$ has $|V(G')|=\tilde O(\alpha)$ nodes, for some parameter $\alpha$, which will be fixed later. Note that the shortest path diameter ($\spdiam(G, w)$) of $(G', w')$  might be as large as $|V(G')|=\tilde O(\alpha)$.\danupon{Should recall ``shortest path diameter''} Since the running time of our algorithm for single-source shortest path will depend on the shortest path diameter, we wish to reduce the shortest path diameter. We will apply the technique from \Cref{sec:spd reduction} to do this task, as follows.

\begin{proposition}[\spdiam reduction of an overlay network]\label{thm:spd reduction overlay network}
For any parameter $\alpha$ and $\beta$, consider an overlay network $(G', w')$ of $\tilde O(\alpha)$ nodes, embedded in network $(G, w)$. There is a distributed algorithm that terminates in $\tilde O(\alpha\beta+\diam(G))$ time and gives an overlay network $(G'', w'')$ such that $V(G')=V(G'')$, $\spdiam(G'', w'') = \tilde O(\alpha/\beta)$, and for any nodes $u$ and $v$, $\dist_{G'', w''}(u, v)=\dist_{G', w'}(u, v)$. 
\end{proposition}
\begin{proof}
Consider the following algorithm. First, every node in the overlay network $(G', w')$ broadcasts to all other nodes the values of $\beta$ edges incident to it with smallest weights (breaking ties arbitrarily). This step takes $\tilde O(\alpha \beta)$ time since there are $\alpha\beta$ edges broadcasted. Using these broadcasted edges, every node $v$ can compute $\beta$ nodes nearest to it (since any shortest path algorithm -- Disjkstra's algorithm for example -- will only need to know $\beta$ smallest-weight edges to compute $\beta$ nearest nodes).  Thus, $v$ can add $\beta$ {\em shortcuts} to the network $(G', w')$ to construct network $(G'', w'')$. In fact, the added shortcuts could be broadcasted to all nodes in $\tilde O(\alpha\beta+\diam(G))$ time since each node will broadcast only $\beta$ shortcuts. This implies that we can build an overlay network $(G'', w'')$ in $\tilde O(\alpha\beta)$ time and, by \Cref{thm:spd reduction}, the shortest-path diameter of $(G'', w'')$ is $\spdiam(G'', w'') = \tilde O(\alpha/\beta).$\danupon{This proof is really rough.}
\end{proof}

\subsection{Computing \sssp on Overlay Network $(G'', w'')$}\label{sec:sssp overlay}

In the final step of our sublinear-time \sssp algorithm, we solve \sssp on overlay network $(G'', w'')$ embedded in $(G, w)$ obtained in the previous section. Recall that for parameters $\alpha$ and $\beta$ which will be fixed later, $|V(G'')|=\tilde \Theta(\alpha)$ and  $\spdiam(G'', w'')=\tilde O(\alpha/\beta)$.

\begin{lemma}[$(1+o(1))$-approximate \sssp on $(G'', w'')$]\label{thm:final step sssp}
We can $(1+o(1))$-approximate \sssp on $(G'', w'')$ in $\tilde O(\diam(G)\alpha/\beta + \alpha)$ time.
\end{lemma}
\begin{proof}
We will simulate the light-weight $h$-hop \sssp algorithm in \Cref{thm:light-weight sssp} on the overlay network $(G'', w'')$ by using $h=\spdiam(G'', w'') = \tilde O(\alpha/\beta)$. To simulate this algorithm, we will view $(G'', w'')$ as a {\em fully-connected} overlay network where every node can communicate with other nodes by {\em broadcasting}, i.e. sending a message to every node in the original network $G$, which takes $O(\diam(G))$ time. In particular, every node in $(G'', w'')$ will simulate each round of this algorithm and wait until the messages that are sent in such round by all nodes are received by all nodes before starting the next round (see \Cref{algo:simulate broadcast on overlay}). 

\begin{algorithm}
\caption{Similating a broadcasting algorithm on an overlay network $(G'', w'')$}\label{algo:simulate broadcast on overlay}
{\bf Input:} An overlay network $(G'', w'')$ embedded on network $G$ and an algorithm $\cA$ such that nodes communicate only by broadcasting a message to all its neighbors. 

{\bf Goal:} Simulate $\cA$ on $(G'', w'')$ when we view $G''$ as a fully-connected overlay network.

\begin{algorithmic}[1]
\For{each round $i$ of algorithm $\cA$} 

\State Count the number of nodes in $G''$ that want to broadcast a message in this round of $\cA$. Let $M_i$ be such number. Make every node in $G$ knows $M_i$. (This step takes $O(\diam(G))$ time.)

\State Every node in $G''$ that wants to send a message broadcasts such message to every node in $G$. Wait for $\diam(G)+M_i$ rounds to make sure that every node receives all $M_i$ messages before proceeding to round $i+1$.

\EndFor
\end{algorithmic}
\end{algorithm}

Simulating each round $i$ of this algorithm will take $\tilde O(\diam(G)+M_i)$, where $M_i$ is the total number of messages broadcasted by all nodes in round $i$. This is because broadcasting $M_i$ messages to all nodes in the network (not just all neighbors) takes  $\tilde O(\diam(G)+M_i)$ time. 
Note that, by \Cref{thm:light}, this algorithm finishes in $\tilde O(h)$ rounds; thus, the total time needed to simulate this algorithm is  $\tilde O(h\diam(G)+M)$ where $M$ is the total number of messages broadcasted by all nodes throughout the algorithm. Since this algorithm is light-weight, every node in $G''$ broadcasts only $O(\log n)$ messages, and thus we can bound $M$ by $\tilde O(|V(G'')|) = \tilde O(\alpha)$. So, the total running time is $\tilde O(h\diam(G)+\alpha) = \tilde O(\diam(G)\alpha/\beta + \alpha)$ as claimed. 
\end{proof}

\subsection{Putting Everything Together (Proof of \Cref{thm:main sssp})}\label{sec:sssp time analysis}

By \Cref{thm:final step sssp}, we can $(1+o(1))$-approximate \sssp on $(G'', w'')$ which, in turn,  $(1+o(1))$-approximates \sssp on $(G', w')$, by \Cref{thm:spd reduction overlay network}. Then, by \Cref{thm:overlay network}, we know that we can $(1+o(1))$-approximate \sssp on the original network $(G, w)$ as desired. 
We now analyze the running time. Constructing $(G', w')$ takes $\tilde O(\alpha+n/\alpha+\diam(G))$, as in \Cref{thm:overlay network}. Adding shortcuts to $(G', w')$ to construct $(G'', w'')$ takes $\tilde O(\alpha\beta+\diam(G))$, by \Cref{thm:spd reduction overlay network}. Finally, solving \sssp on $(G'', w'')$ takes $\tilde O(\diam(G)\alpha/\beta + \alpha)$ by \Cref{thm:final step sssp}. So, the total running time of our algorithm is 
\[
\tilde O(n/\alpha+\diam(G)+\alpha\beta+\diam(G)\alpha/\beta).
\]
By setting $\alpha=n^{1/2}/(\diam(G))^{1/4}$ and $\beta=(\diam(G))^{1/2}$, we get the running time of $\tilde O(n^{1/2}(\diam(G))^{1/4}+\diam(G))$ as desired. 
Note that it is possible that $\beta\geq \alpha$. In this case, we will simply set $\beta=\alpha$ to get the claimed running time; in fact, this happens only when $\diam(G)\geq n^{2/3}$, and the running time will be $\tilde O(\diam(G))$ in this case. 


\section{Algorithms on Fully-Connected Networks}\label{sec:clique}

\subsection{$\tilde O(\sqrt{n})$-time Exact Algorithm for \sssp}\label{sec:sssp clique}

In this section, we present an algorithm that solves \sssp exactly in $\tilde O(\sqrt{n})$ time on fully-connected networks. The algorithm has two simple phases, as shown in \Cref{algo:clique SSSP}. In the first phase, it reduces the shortest-path diameter using the techniques developed in \Cref{sec:spd reduction}. In particular, every node $u$ broadcasts $k=\sqrt{n}$ edges of smallest weight. Then, every node uses the information it receives to compute a $k$-shortcut graph $(G, w')$, which can be done due to \Cref{thm:shortcuts from smallest edges}. By \Cref{thm:spd reduction}, we have 
$$\spdiam(G, w')<4\sqrt{n}~~~\mbox{and}~~~\forall u, v:~\dist_{G, w}(u, v)=\dist_{G, w'}(u, v).$$

\begin{algorithm}
\caption{$\tilde O(\sqrt{n})$-time Exact Algorithm for \sssp}\label{algo:clique SSSP}
{\bf Input:} A fully connected network $(G, w)$ and source node $s$. Weight $w(uv)$ of each edge $uv$ is known to $u$ and $v$. 

{\bf Output:} Every node $u$ knows $d(s, u)$ which is the equal to $\dist_{G, w}(s, u)$. 

\medskip{\bf Phase 1:} Shortest path diameter reduction. This phase gives a new weight $w'$ such that $\spdiam(G, w')<4\sqrt{n}$. The weight $w'(uv)$ of an edge $uv$ is known to its end-nodes $u$ and $v$.

\medskip\begin{algorithmic}[1]
\State Let $k=\sqrt{n}$. 
\State Each node $u$ sends $k$ edges of smallest weight, i.e. edges in $E^k(u)$ as in \Cref{def:E^k(u) set of k smallest weight }, to all other nodes. 
\State Every node $v$ uses $\bigcup_{v\in V(G)} E^k(v)$ construct  $(G^k, w)$ and compute $k$-shortcut edges, i.e. compute  $S^k_{G, w}(u)$ and $\{\dist_{G, w}(u, v)\}_{v\in S^k_{G, w}(u)}$. {\footnotesize\em // This step can be done internally (without communication) due to \Cref{thm:shortcuts from smallest edges}}. 
\State Augment $(G, w)$ with $k$-shortcut edges: for any edge $uv$, let $w'(uv)=\dist_{G, w}(u, v)$ if $u\in S^k_{G, w}(v)$ or $v\in S^k_{G, w}(u)$; otherwise, $w'(uv)=w(uv)$. {\footnotesize\em // This step can be done internally since both $u$ and $v$ know all information needed, i.e. $S^k_{G, w}(u)$,  $S^k_{G, w}(v)$,  $\{\dist_{G, w}(u, v)\}_{v\in S^k_{G, w}(u)}$, and $\{\dist_{G, w}(u, v)\}_{u\in S^k_{G, w}(v)}$.}
\algstore{cliqueSSSP}
\end{algorithmic}

\medskip{\bf Phrase 2:} Simulate Bellman-Ford algorithm on $(G, w')$. This phase makes every node $u$ knows  $d(s, u)$ where we claim that $d(s, u)=\dist_{G, w}(s, u)$. 
\medskip\begin{algorithmic}[1]
\algrestore{cliqueSSSP}
\State Let $d(s, s)=0$ and $d(s, u)=\infty$ for every node $u$.
\For{$i=1\ldots 4\sqrt{n}$}
\State Every node $u$ sends $d(s, u)$ to all other nodes. 
\State Every node $v$ updates $d(s, v)$ to $\min_{u} (d(s, u)+w'(uv))$. 
\EndFor
\end{algorithmic}
\end{algorithm}

In the second phase, the algorithm simulates Bellman-Ford's algorithm on $(G, w')$. In particular, every node iteratively uses the distance from $s$ to other nodes to update its distance; i.e., every node $v$ sets $d(s, v)$ to $\min_{u} (d(s, u)+w'(uv))$. It can be easily shown that by repeating this process for $\spdiam(G, w')$ iteration,  $d(s, u)=\dist_{G, w'}(s, u)$ for every node $u$. We provide the sketch of this claim for completeness, as follows.\danupon{To do: Use $d(u)$ instead of $d(s, u)$.}
\begin{claim}[Correctness of Phase 2 of \Cref{algo:clique SSSP}]\label{thm:BellmanFord}
Phase~2 of \Cref{algo:clique SSSP} returns a function $d$ such that, for every node $u$, $d(s, u)=\dist_{G, w'}(s, u)$. 
\end{claim}
\begin{proof}
We will show by induction that after the $i^{th}$ iteration, the value of $d(s, u)$ will be at most the value of the $i$-hop distance between $s$ and $u$, i.e. $d(s, u)\leq \dist_{G, w'}^i(s, u)$ (recall that $\dist_{G, w'}^i(s, u)$ is defined in \Cref{def:bounded hop sssp}). This trivially holds before we start the first iteration since $\dist_{G, w'}^0(s, s)=0$ and $\dist_{G, w'}^0(s, u)=\infty$. Assume for an induction that it holds for some $i\geq 0$. For any node $u$, let $P$ be a shortest $(i+1)$-hop $s$-$u$ path, and $v$ be the node preceding $u$ in such path. By the induction hypothesis, after the $i^{th}$ iteration, $d(s, v)\leq \dist^i_{G, w'}(s, v)$. So, after the $(i+1)^{th}$ iteration, $d(s, u)\leq d(s, v)+w'(vu)\leq \dist_{G, w'}(s, u)$. The claim thus holds for the $(i+1)^{th}$ iteration.

Let $h=\spdiam(G, w')$. Since $\dist^h_{G, w'}(s, u)=\dist_{G, w'}(s, u)$ for every node $u$, we have that $d(s, u)\leq \dist_{G, w'}(s, u)$ after $h$ iterations. Since it is clear that $d(s, u)\geq \dist_{G, w'}$, \Cref{thm:BellmanFord} follows.
\end{proof}

\subsection{$\tilde O(\sqrt{n})$-time $(2+o(1))$-Approximation Algorithm for \apsp}\label{sec:apsp clique}

We now present a $(2+o(1))$-approximation algorithm for \apsp, which also has $\tilde O(\sqrt{n})$ time. Our algorithm is outlined in \Cref{algo:clique APSP}. In the first phase of this algorithm is almost the same as the first phase of \Cref{algo:clique SSSP} presented in the previous section: by having every node $u$ sending out $E^{\sqrt{n}}(u)$ to all other nodes, we get a network $(G, w')$ such that 
$$\spdiam(G, w')<4\sqrt{n}~~~\mbox{and}~~~\forall u, v:~\dist_{G, w}(u, v)=\dist_{G, w'}(u, v).$$ 
The only difference is that, in addition to performing Phase~1 of \Cref{algo:clique SSSP} to get the properties above, we also make sure that 
\begin{align}
w'(uv)\leq \min_{z\in S^k(u)}\dist_{G, w}(u, z)+w(zv)\,.\label{eq:cliqueAPSP one}
\end{align}
This is done by having every node $u$ broadcasts $S^k(u)$ and $\{\dist_{G, w}(u, z)\}_{z\in S^k_{G, w}(u)}$ (which are also computed in Phase~1 of \Cref{algo:clique SSSP}) to all other nodes, where $k=\sqrt{n}$. Then every node $v$ can internally update $w'(uv)$, for every node $u$, to $w'(uv)= \min\{w'(uv), \min_{z\in S^k(u)}\dist_{G, w}(u, z)+w(zv)\}$. 
Phase~1 takes $\tilde O(\sqrt{n})$ time since performing Phase~1 of \Cref{algo:clique SSSP} takes $\tilde O(\sqrt{n})$ time and broadcasting $S^k(u)$ and $\{\dist_{G, w}(u, z)\}_{z\in S^k_{G, w}(u)}$, which are sets of size  $\tilde O(\sqrt{n})$, also takes  $\tilde O(\sqrt{n})$ time.

\begin{algorithm}
\caption{$\tilde O(\sqrt{n})$-time $(2+o(1))$-Approximation Algorithm for \apsp}\label{algo:clique APSP}
{\bf Input:} A fully connected network $(G, w)$. Weight $w(uv)$ of each edge $uv$ is known to $u$ and $v$. 

{\bf Output:} Every node $u$ knows $(2+o(1))$-approximate value of $\dist_{G, w}(u, v)$ for every node $v$.

%

\medskip{\bf Phase 1:} Let $k=\sqrt{n}$. Compute $S^{k}(u)$, for every node $u$, and weight assignment $w'$ such that  $\spdiam(G, w')<4\sqrt{n}$ and $w'(uv)\leq \min_{z\in S^k(u)}\dist_{G, w}(u, z)+w(zv)$.

\medskip\begin{algorithmic}[1]
\State Perform Phase~1 of \Cref{algo:clique SSSP}. This step makes every node $u$ knows $S^k(u)$ and $\{\dist_{G, w}(u, z)\}_{z\in S^k_{G, w}(u)}$. 

\State Every node $u$ sends $S^k(u)$ and $\{\dist_{G, w}(u, z)\}_{z\in S^k_{G, w}(u)}$ to all other nodes. 

\State Every node $v$ updates $w'(uv)$, for every node $u$, to $w'(uv)= \min\{w'(uv), \min_{z\in S^k(u)}\dist_{G, w}(u, z)+w(zv)\}$. 
{\footnotesize\em // This step can be done without communication since every node $v$ knows  $S^k(u)$ and $\{\dist_{G, w}(u, z)\}_{z\in S^k_{G, w}(u)}$  for all nodes $u$}
\algstore{cliqueAPSP}
\end{algorithmic}

\medskip{\bf Phase 2:} Compute $(4\sqrt{n})$-hop $(\sqrt{n}\log n)$-source shortest paths for $\sqrt{n}\log n$ random sources.

\medskip\begin{algorithmic}[1]
\algrestore{cliqueAPSP}

\State Let $R$ be a set of randomly selected $\sqrt{n}\log n$ nodes. 

\State Run the multi-souce bounded-hop shortest paths algorithm from \Cref{thm:k-source h-hop sp} (\Cref{algo:bounded-hop multi-source}) on $(G, w')$ for $h=4\sqrt{n}$ hops using nodes in $R$ as sources. 
{\footnotesize // \em At the end of this process, every node $v$ knows a $(1+o(1))$-approximate value of $\dist^h_{G, w'}(r, v)$, which equals to $\dist_{G, w'}(r, v)$ since  $\spdiam(G, w')<4\sqrt{n}$, for all $r\in R$. We denote this approximate value by $d'(u, v)$; thus, $\dist_{G, w'}(r, v)\leq d'(u, v)\leq (1+o(1))\dist_{G, w'}(r, v)$.}

\State Every node $v$ sends $d'(r, v)$, for all $r\in R$, to all nodes. 
\end{algorithmic}

\medskip{\bf Final Phase:} Every node $v$ uses the information it knows so far (see \Cref{def:G_u}) to compute $d''(u, v)$ for all nodes $u$, which is claimed to be a $(2+o(1))$ approximation of $\dist_{G, w}(u, v)$ (see \Cref{thm:cliqueAPSP}). 

\end{algorithm}

In the second phase, we pick $\Theta(\sqrt{n}\log n)$ nodes uniformly at random. Let $R$ be the set of these random nodes. We run the light-weight $h$-hop $t$-source \sssp algorithm (\Cref{algo:bounded-hop multi-source}) from these random nodes using $h=\sqrt{n}$ and $t=|R|$. By \Cref{thm:k-source h-hop sp}, we will finish in $\tilde O(|R|+h) = \tilde O(\sqrt{n})$ rounds with high probability. Moreover, since the shortest-path distance is reduced to $\sqrt{n}$, every node will know an $(1+o(1))$-approximate distance to {\em all} random nodes in $R$. Every node broadcasts these distances to nodes in $R$ to all other nodes. This takes $O(|R|)$ time since the network is fully connected.
In the final phase, every node $u$ uses these broadcasted distances and the distances it computes in the previous step (by simulating Dijkstra's algorithm) to compute the approximate distance between itself and other nodes. 
In particular, for any node $u$, consider the following graph $(G_u, w_u)$.  

\begin{definition}[Graph $(G_u, w_u)$]\label{def:G_u} Graph $(G_u, w_u)$ consists of the following edges. 
\begin{enumerate}
\item edges from $u$ to all other nodes $v$ of weight $w'(u, v)$, 
\item edges from $x\neq u$ to nodes $y\in S^k(x)$ of weight $w'(x, y)=\dist_{G, w}(x, y)$, and 
\item edges from every random node $r\in R$ to all other nodes $v$ of weight $d''(r, v)$. (Recall that $\dist_{G, w}(r, v)\leq d''(r, v)\leq (1+o(1))\dist_{G, w}(r, v)$.) 
\end{enumerate}
\end{definition}
Node $u$ will use $\dist_{G_u, w_u}(u, v)$ as an approximate distance of $\dist_{G_u, w}(u, v)$. We now show that this gives a $(2+o(1))$-approximate distance. 

\begin{lemma}[Approximation guarantee of \Cref{algo:clique APSP}]\label{thm:cliqueAPSP}
For every pair of nodes $u$ and $v$, $\dist_{G, w}(u, v)\leq \dist_{G_u, w_u}(u, v)\leq (2+o(1))\dist_{G, w}(u, v).$
\end{lemma}
\begin{proof}
It is clear that $\dist_{G, w}(u, v)\leq \dist_{G_u, w_u}(u, v)$ since $w_u(xy)\geq \dist_{G, w}(x, y)$ for every pair of nodes $x$ and $y$. It is left to prove that $\dist_{G_u, w_u}(u, v)\leq (2+o(1))\dist_{G, w}(u, v).$
Let 
$$P=\langle u=x_1, x_2, \ldots, x_k=v\rangle$$  
be a shortest $u$-$v$ path. Note that we can assume that $k\leq \sqrt{n}$ since the shortest-path diameter is $\sqrt{n}$. 
Let $x_i$ be the furthest node from $u$ that is in $S^k(u)\cap P$ and, similarly,  let $x_j$ be the furthest node from $v$ that is in $S^k(v)\cap P$; i.e. 
$$i=\arg\max_{i'} (x_{i'}\in S^k(u)\cap P) ~~~\mbox{and}~~~ j=\arg\min_{j'} (x_{j'}\in S^k(v)\cap P) \,.$$
Note that $x_1, \ldots, x_i$ are all in $S^k(u)$ since all nodes $x_1, \ldots, x_{i-1}$ are nearer to $u$ than $x_i$. Similarly, $x_j, \ldots, x_k$ are all in $S^k(v)$. Note further that $w'(u, x_{i+1})=\dist_{G, w}(u, x_{i+1})$ since Phase~1 guarantees \Cref{eq:cliqueAPSP one} which implies that
\begin{align*}
w'(u, x_{i+1}) &\leq \min_{z\in S^k(u)} \dist_{G, w}(u,z)+w(zx_{i+1})\\
&\leq \dist_{G, w}(u, x_i)+w(x_ix_{i+1})  &\mbox{(since $x_i\in S^k(u)$)}\\
&= \dist_{G, w}(u, x_{i+1})\,.
\end{align*}
\begin{figure}
\centering
\begin{subfigure}[b]{0.45\textwidth}
\centering
\includegraphics[page=1, clip=true, trim=0.95cm 10.5cm 12.3cm 2.2cm, width=\textwidth]{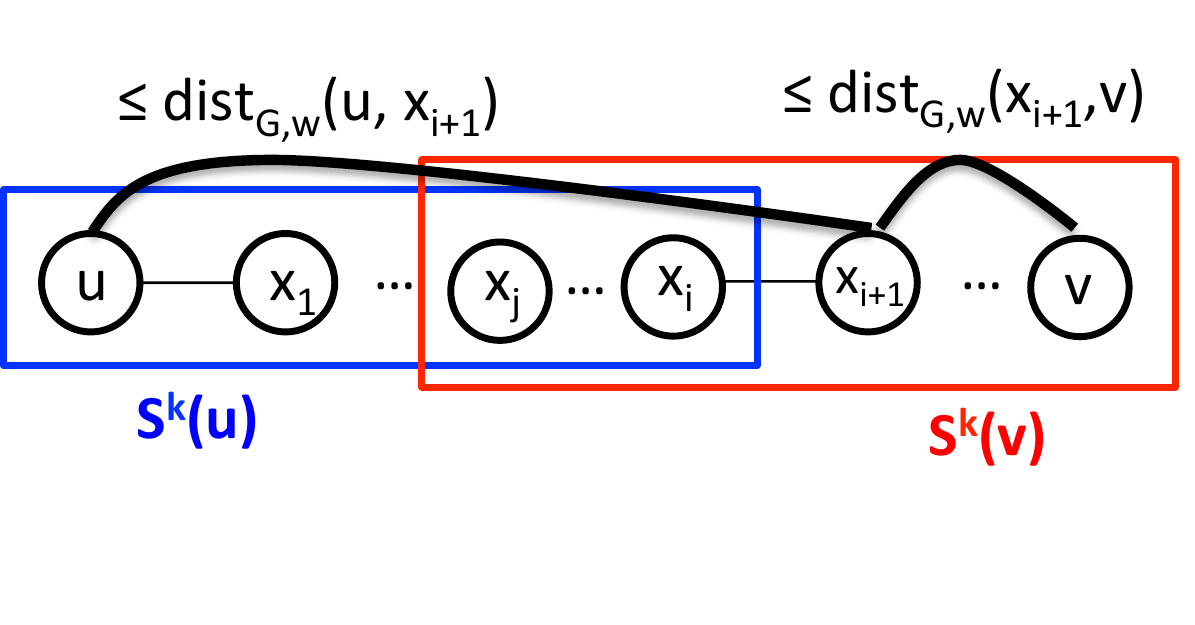}
\caption{Case 1: $j\leq i+1$}
\label{fig:case1}
\end{subfigure}
\begin{subfigure}[b]{0.45\textwidth}
\centering
\includegraphics[page=1, clip=true, trim=0.95cm 4cm 12.25cm 8cm, width=\textwidth]{pictures.pdf}
\caption{Case 2: $j> i+1$.}
\label{fig:case2}
\end{subfigure}
\caption{Outline of the proof of \Cref{thm:cliqueAPSP}}\label{fig:cliqueAPSP}
\end{figure}
We now consider two cases (see \Cref{fig:cliqueAPSP} for an outline). First, if $j\leq i+1$, then we have $x_{i+1}\in S^k(u)\cap S^k(v)$. This means that $w_u(x_{i+1}, v)=w'(x_{i+1}, v)=\dist_{G, w}(x_{i+1}, v)$. 
\begin{align*}
\dist_{G_u, w_u}(u, v) & \leq w_u(u, x_{i+1})+w_u(x_{i+1}, v)\\ 
& = w'(u, x_{i+1}) + w'(x_{i+1}, v)\\
& = \dist_{G, w}(u, x_{i+1}) + \dist_{G, w}(x_{i+1}, v)\\
& = \dist_{G, w}(u, v).
\end{align*}
%
Now we consider the second case where $j>i+1$. 
%
%
In this case, we have that either
$$\dist_{G, w}(u, x_{i+1})\leq \dist_{G, w}(u, v)/2 ~~~\mbox{or}~~~ \dist_{G, w}(v, x_{j-1})\leq \dist_{G, w}(u, v)/2\,.$$  
Since the analyses for both cases are essentially the same, we will only show the analysis when $\dist_{G, w}(u, x_{i+1})\leq \dist_{G, w}(u, v)/2$.
Observe that, with high probability, there is a random node $r\in R$ that is in $S^k(u)$ since $|S^k(u)|\geq \sqrt{n}$ and $R$ consists of $\tilde \Theta(\sqrt{n})$ random nodes (see \Cref{thm:UllmanY}). Recall that, for every node $z$, 
$$w_u(r, z)\leq d''(r, z) \leq (1+o(1))\dist_{G, w}(r, z).$$ 
It follows that
\begin{align}
\dist_{G_u, w_u}(u, v) & \leq w_u(u, r)+w_u(r, v)\label{eq:cliqueAPSP1}\\ 
& \leq (1+o(1))(\dist_{G, w}(u, r) + \dist_{G, w}(r, v))\label{eq:cliqueAPSP2}\\
& \leq (1+o(1))(\dist_{G, w}(u, r) + (\dist_{G, w}(u, r)+\dist_{G, w}(u, v)))\label{eq:cliqueAPSP3}\\
& \leq (1+o(1))(2\dist_{G, w}(u, x_{i+1}) + \dist_{G, w}(u, v))\label{eq:cliqueAPSP4}\\
& \leq (2+o(1)) \dist_{G, w}(u, v).\label{eq:cliqueAPSP5}
\end{align}
\Cref{eq:cliqueAPSP3} is by triangle inequality. \Cref{eq:cliqueAPSP4} is because $r\in S_u$ and $x_{i+1}\notin S_u$. \Cref{eq:cliqueAPSP5} is because of the assumption that $\dist_{G, w}(u, x_{i+1})\leq \dist_{G, w}(u, v)/2$. 
\end{proof}

\section{Lower Bound for Approximating APSP (Proof of \Cref{observe:lower bound})}\label{sec:lower bound}


\lowerbound*

\begin{proof}
Our proof simply formalizes the fact that a node needs to receive at least $n$ bits of information in order to know its distance to all other nodes. We start from the following {\em messag sending problem}: Alice receive a $\beta$-bit binary vector, denoted by $\langle x_1, \ldots, x_\beta\rangle$, where we set $\beta=n-2$. She wants to send this vector to Bob. Intuitively, to be sure that Bob gets the value of the vector correctly with a good probability, i.e. with probability at least $1-\epsilon$ for some small $\epsilon>0$, Alice has to send $\Omega(\beta)$ bits to Bob, regardless of what Bob sends to her. This fact can be formally proved in many ways (e.g., by using communication complexity lower bounds)
%
and is true even in the quantum setting (see, e.g., Holevo's theorem \cite{Holevo73}). 

Now, let $\cA$ be an $\alpha(n)$-approximation $T$-time algorithm for weighted \apsp. We show that Alice can use $\cA$ to send her message to Bob using $O(T\log n)$ bits, as follows. Construct a graph $G$ consisting of $n=\beta+2$ nodes, denoted by $a_1, \ldots, a_\beta$, $a^*$ and $b$. There are edges between all nodes to $a^*$.  The weight of edge $a^*b$ is always $w(a^*b)=1$. Weight of every edge $a_ia^*$ is set by Alice: if $x_i=1$ then she sets weight of $a_ia^*$ to $w(a_ia^*)=1$; otherwise she sets it to $w(a_ia^*)=2\alpha(n)$. Then, Alice simulates $\cA$ on $a_1, \ldots, a_\beta$ and $a^*$, and Bob simulates $\cA$ on $b$. If $\cA$ wants to send any message from $a^*$ to $b$, Alice will send this message to Bob so that Bob can continue simulating $b$. Similarly, if  $\cA$ wants to send any message from $b$ to $a^*$, Bob will send this message to Alice so that Alice can continue simulating $a^*$.
If $\cA$ finishes in $T$ rounds, then Alice will send at most $O(T\log n)$ bits to Bob in total. 

Observe that, for any $i$, if $x_i=1$ then $\dist_{G, w}(a_i, b)=2$; otherwise, $\dist_{G, w}(a_i, b)=2\alpha(n)+1$. Since $\cA$ is $\alpha(n)$ approximation, $\cA$ must answer $\tilde \dist_{G, w}(a_i, b)\leq 2\alpha(n)$ if $x_i=1$; otherwise, $\cA$ must answer $\tilde \dist_{G, w}(a_i, b)\geq 2\alpha(n)+1$. Since $\tilde \dist_{G, w}(a_i, b)$ is known to $b$, Bob can get the value of $\tilde \dist_{G, w}(a_i, b)$ by reading it from $b$ (which he is simulating). Then he can reconstructs $x_i$. Thus, Bob can reconstruct all bits $x_1, \ldots, x_\beta$ after getting $O(T\log n)$ bits from Alice.
The lower bound of the message sending problem thus implies that $T=\Omega(\beta/\log n)=\Omega(n/\log n)$. 

Note that since the highest weight we can put on an edge is $\poly(n)$, we require that $\alpha(n)\leq \poly(n)$. 
We use the same argument for the unweighted case, but this time we use $\beta=n/\alpha(n)$ and replace an edge of weight $2\alpha(n)$ by a path of length $2\alpha(n)$. 
\end{proof}

Note that in the proof above we show a lower bound for computing distances between all pairs of nodes. Since the lower bound graph is a star, the routing problem is trivial (since there is always one option to send a message). We can easily modify the above graph to give the same lower bound for the routing problem on weighted graphs: First, instead of using weight $2\alpha(n)$ in the graph above, use weight $2\alpha^2(n)+\alpha(n)$ instead. Second, add a new node $c$ and edges of weight $2\alpha(n)$ between $c$ and all nodes $a_1$ and an edge of weight $1$ between $c$ and $b$. Observe that if $x_i=1$, then we have to route a message through $a^*$, giving a distance of $2$ (while routing through $c$ gives a distance of $2\alpha(n)+1$. If $x_i=0$, we should route through $c$ which gives a distance of $2\alpha(n)+1$ since routing through $a^*$ will cost $2\alpha^2(n)+\alpha(n)+1$.

\danupon{To do: Mention that proving lower bound on clique is quite impossible (by Oshman's talk)}

\section{Open Problems}


The main question left by our \sssp algorithm is the following. 
\begin{problem}
Close the gap between the upper bound of $\tilde O(n^{1/2}\diam^{1/4})$ presented in this paper and the lower bound of $\tilde \Omega(n^{1/2})$ presented in \cite{DasSarmaHKKNPPW11} for $(1+\epsilon)$-approximating the single-source shortest paths problem on general networks. 
\end{problem}

Improving the current upper bound is important since there are many problems that can be potentially solved by using the same technique. Moreover, giving a lower bound in the form  $\tilde \Omega(n^{1/2}\diam^{\delta})$ for some $\delta>0$ will be quite surprising since such lower bound has not been observed before. 
%
%
It should also be fairly interesting to refine our upper bound to achieve a $\tilde O(n^{1/2}\diam^{\epsilon})$-time $O(1/\epsilon)$-approximation algorithm for any $\delta>0$. 
Another question that should be very interesting is understanding the exact case: 
\begin{problem}
Can we solve \sssp {\em exactly} in sublinear-time? 
\end{problem}
It is also interesting to solve \apsp {\em exactly} in linear-time (recall that sublinear-time is not possible). In some settings, an exact algorithm for computing shortest paths is crucial; e.g. some Internet protocols such as OSPF and IS-IS use edge weights to control the traffic and using an approximate shortest paths with this protocol is unacceptable\footnote{We thank Mikkel Thorup for pointing out this fact}. 
%
%
The next question is a generalization of our \sssp: 
\begin{problem}[Asymmetric \sssp]
How fast can we solve \sssp on networks whose edge weights could be {\em asymmetric}, i.e. if we think of each edge $uv$ as two directed edges $\overrightarrow{uv}$ and $\overrightarrow{vu}$, it is possible that $w(\overrightarrow{uv})\neq w(\overrightarrow{vu})$. 
\end{problem}
Note that we are particularly interested in the case where weights do not affect communication; in other words, if $u$ can send a message to $v$, then $v$ can also send a message to $u$. Also note that our light-weight \sssp algorithm can be used to solve this problem (but not the shortest-path diameter reduction technique). By adjusting parameters appropriately, we can $(1+\epsilon)$-approximate this problem in $\tilde O(\min(n^{2/3}, n^{1/2}D^{1/2}))$ time. 
In fact, improving this running time for the following very special case seems challenging already:
\begin{problem}[$s$-$t$ Reachability Problem]
Given a directed graph $G$ and two special nodes $s$ and $t$, we want to know whether there is a directed path from $s$ to $t$. The communication network is the underlying undirected graph; i.e. the communication can be done independent of edge directions and the diameter $\diam$ is defined to be the diameter of the underlying undirected graph. Can we answer this question in $\tilde O(\sqrt{n}+\diam)$ time?
\end{problem}
%
This problem shows limitations of the techniques presented in this paper, and we believe that solving it will give a new insight into solving all above open problems. Our last set of questions: 
\begin{problem}
Can we improve the $\tilde O(n^{1/2})$-time upper bound for \sssp on fully-connected networks while keeping the approximation ratio small (say, at most two)? Is it possible to prove a nontrivial $\omega(1)$ lower bound?
\end{problem}
Note that the last question was asked earlier by Elkin \cite{Elkin04}. 


\section{Acknowledgement}
The sublinear-time \sssp algorithm was  inspired by the discussions with Jittat Fakcharoenphol and Jakarin Chawachat, who refused to co-author this paper. Several techniques borrowed from dynamic graph algorithms benefit from many intensive discussions with Sebastian Krinninger and Monika Henzinger.
I also would like to thank  Radhika Arava and Peter Robinson for explaining the algorithm of Baswana and Sen \cite{BaswanaS07}\danupon{Correct?} to him and Gopal Pandurangan, Chunming Li, Anisur Rahaman, David Peleg, Atish Das Sarma, Parinya Chalermsook, Bundit Laekhanukit, Boaz Patt-Shamir, Shay Kutten, Christoph Lenzen, Stephan Holzer, Mohsen Ghaffari, Nancy Lynch, and Mikkel Thorup, for discussions, comments, and pointers to related results. I also thank all reviewers of STOC 2014 for many thoughtful comments.

\bibliographystyle{alpha}
\bibliography{references-APSP}  


%
%
%
%

\end{document}